\documentclass{lmcs}

\usepackage[toc, page]{appendix}
\usepackage{graphicx}
\usepackage{tipa}
\usepackage{bbm}
\usepackage{etex}
\usepackage{setspace}
\usepackage[most]{tcolorbox}
\usepackage{amsmath}
\usepackage{stackengine}
\usepackage{amssymb}
\usepackage{amsfonts}
\usepackage{mathtools}
\usepackage{thmtools}
\usepackage{stmaryrd}
\usepackage{bussproofs}
\usepackage{caption}
\makeatletter
\let\th@plain\relax
\usepackage{tikz}
\usetikzlibrary{arrows}
\usepackage{tikz-cd}
\usepackage{Common}
\usepackage{xcolor}
\usepackage{hyperref}
\hypersetup{
    colorlinks,
    citecolor=green,
    filecolor=blue,
    linkcolor=blue,
    urlcolor=brown}

\renewcommand{\epsilon}[0]{\varepsilon}

\newcommand{\todo}[1]{\red{[TODO: #1]}}

\newcommand{\ds}[1]{\widehat {#1}}

\newcommand{\CAct}[0]{CAct}

\newcommand{\cart}[0]{\cat{Cat_\times}}
\newcommand{\ftor}[1]{\mathrm{#1}}
\newcommand{\A}[1]{\mathbf{#1}}
\newcommand{\theo}[1]{\mathfrak{T}_{#1}}
\newcommand{\CSet}[0]{\cat{Set}}

\newcommand{\CGrpS}[0]{\cat{Grp}_{\CSet}}

\newcommand{\arr}[1]{\mathrm{#1}}
\renewcommand\Id{\arr{Id}}

\makeatletter
\newsavebox{\@brx}
\newcommand{\llangle}[1][]{\savebox{\@brx}{\(\m@th{#1\langle}\)}%
  \mathopen{\copy\@brx\kern-0.5\wd\@brx\usebox{\@brx}}}
\newcommand{\rrangle}[1][]{\savebox{\@brx}{\(\m@th{#1\rangle}\)}%
  \mathclose{\copy\@brx\kern-0.5\wd\@brx\usebox{\@brx}}}
\makeatother
\newcommand{\kpair}[2]{\llangle {#1}, {#2} \rrangle}

\newcommand{\reachOrder}[0]{\sqsubseteq}
\newcommand{\changes}[0]{\Delta}
\newcommand{\change}[0]{\delta}
\newcommand{\plus}[0]{+}
\renewcommand{\d}[0]{\partial}
\newcommand{\der}[2]{\frac{\d\,{#1}}{\d{#2}}}

\newcommand{\vf}[1]{\mathcal{#1}}

\newcommand{\pr}[0]{{\sf p}}

\newcommand{\da}[1]{#1} 
\newcommand{\us}[1]{{|{#1}|}}

\newcommand{\KK}[0]{\mathbb{K}}

\newcommand\anglebra[1]{\langle #1 \rangle}
\newcommand\tmap[3]{\boldsymbol{#1}^{({#2},{#3})}}

\newcommand\epow[1]{\pi_1^{{\langle#1\rangle}}}
\newcommand\cpow[1]{\pi_2^{{(#1)}}}

\newcommand{\DF}[1]{\ftor{D}\left[ {#1} \right]}
\newcommand{\ax}[1]{\textbf{[CD.#1]}}
\newcommand{\repeatthm}[1]{%
  \begingroup
  \renewcommand{\thethm}{\ref{thm:#1}}%
  \expandafter\expandafter\expandafter\thm
  \csname reptheorem@#1\endcsname
  \endthm
  \endgroup
}

\NewEnviron{proof-for-theorem}[1]{
  \repeatthm{#1}
  \begin{proof}
    \label{prf:#1}
    \BODY
  \end{proof}
}

\NewEnviron{rep-theorem}[1]{%
  \global\expandafter\xdef\csname reptheorem@#1\endcsname{%
    \unexpanded\expandafter{\BODY}%
  }%
  \expandafter\thm\BODY\unskip\label{thm:#1}\endthm
  \gdef\lasttheorem{#1}
}

\newcommand{\repeatlemma}[1]{%
  \begingroup
  \renewcommand{\thelem}{\ref{lem:#1}}%
  \expandafter\expandafter\expandafter\lem
  \csname replemma@#1\endcsname
  \endlem
  \endgroup
}

\NewEnviron{proof-for-lemma}[1]{
  \repeatlemma{#1}
  \begin{proof}
    \label{prf:#1}
    \BODY
  \end{proof}
}

\NewEnviron{rep-lemma}[1]{%
  \global\expandafter\xdef\csname replemma@#1\endcsname{%
    \unexpanded\expandafter{\BODY}%
  }%
  \expandafter\lem\BODY\unskip\label{lem:#1}\endlem
  \gdef\lasttheorem{#1}
}

\newcommand{\repeatprop}[1]{%
  \begingroup
  \renewcommand{\theprop}{\ref{prop:#1}}%
  \expandafter\expandafter\expandafter\prop
  \csname repproposition@#1\endcsname
  \endprop
  \endgroup
}

\NewEnviron{proof-for-prop}[1]{
  \repeatprop{#1}
  \begin{proof}
    \label{prf:#1}
    \BODY
  \end{proof}
}

\NewEnviron{rep-proposition}[1]{%
  \global\expandafter\xdef\csname repproposition@#1\endcsname{%
    \unexpanded\expandafter{\BODY}%
  }%
  \expandafter\prop\BODY\unskip\label{prop:#1}\endprop
  \gdef\lasttheorem{#1}
}

\newcommand{\proofref}[0]{
  \begin{proof}
    See Appendix~\ref{prf:\lasttheorem}
  \end{proof}
}

\renewcommand\defeq{\coloneqq}

\newif\ifdraft\draftfalse

\ifdraft
\newcommand\lo[1]{{\color{blue} [{#1}--Luke]}}
\newcommand\lochanged[1]{{\color{blue} #1}}
\newcommand\map[1]{{\color{red} [{#1}--Mario]}}

\else
\newcommand\lo[1]{}
\newcommand\lochanged[1]{#1}
\newcommand\map[1]{}

\renewcommand{\todo}[1]{}
\fi

\usepackage[square]{natbib}
\renewcommand\cite[1]{\citep{#1}}

\begin{document}
%
\title{Change Actions: 
Models of Generalised Differentiation}

\author[M.~Alvarez-Picallo]{Mario Alvarez-Picallo}
\address{University of Oxford}
\email{mario.alvarez-picallo@cs.ox.ac.uk}

\author[C.-H.~L.~Ong]{C.-H. Luke Ong}
\address{University of Oxford}
\email{luke.ong@cs.ox.ac.uk}

\begin{abstract}
  Change structures, introduced by Cai et al., have recently been proposed as a semantic framework for incremental computation.
  We generalise change actions, an alternative to change structures, to arbitrary cartesian categories and propose the notion of \emph{change action model} as a categorical model for (higher-order) generalised differentiation.
  Change action models naturally arise from many geometric and computational settings, such as (generalised) cartesian differential categories, group models of discrete calculus, and Kleene algebra of regular expressions.
  We show how to build \lochanged{canonical} change action models on arbitrary cartesian categories, reminiscent of the F\`aa di Bruno construction.
\end{abstract}

\maketitle              

\tableofcontents

\section{Introduction}
\label{sec:intro}
Incremental computation is the process of incrementally updating the output of some given function as the input is gradually changed, without recomputing the entire function from scratch. 
Recently, Cai et al.~\cite{cai2014changes} introduced the notion of change structure to give a semantic account of incremental computation. 
Change structures have subsequently been generalised to \emph{change actions} \cite{mario2019fixing}, 
and proposed as a model for automatic differentiation \cite{kelly2016evolving}.
These developments raise a number of questions about the structure of change actions themselves and how they relate to more traditional notions of  differentiation. 


A \emph{change action} $A = (\us A, \changes A, \oplus_A, +_A, 0)$ is a set $\us A$ equipped with a monoid $(\changes A, +_A, 0_A)$ acting on it, via action $\oplus_A : \us A \times \changes A \to \us A$.
For example, every monoid $(S, +, 0)$ gives rise to a (so-called \emph{monoidal}) change action $(S, S, +, +, 0)$.
Given change actions $A$ and $B$, consider functions $f : \us A \to \us B$.
A \emph{derivative} of $f$ is a function $\d f : \us A \times \changes A \to \changes B$ such that for all $a \in \us A, \change a \in \changes A$,
\(
f (a \oplus_A \change a) = f(a) \oplus_B \d f(a, \change a).
\) 
Change actions and differentiable functions (i.e.~functions that have a regular derivative) organise themselves into categories (and indeed 2-categories) with finite (co)products, whereby morphisms are composed via the chain rule.

The definition of change actions (and derivatives of functions) makes no use of properties of $\cat{Set}$ beyond the existence of products.
We develop the theory of change actions on arbitrary cartesian categories and study their properties. 
A first contribution is the notion of a \emph{change action model}, which is defined to be a coalgebra for a certain (copointed) endofunctor $\ftor{\CAct}$ on the category $\cart$ of (small) cartesian categories.
The functor $\ftor{\CAct}$ sends a category $\cat{C}$ to the category $\ftor{\CAct}(\cat{C})$ of (internal) change actions and differential maps on $\cat{C}$.



\lo{What are the benefits of showing that something is a change action model? What general properties follow? 1. Chain rule. 2. Lemma 5 (Internalisation of change action). 3. Structure of products: partial derivatives make sense) 4. Theorem 5.}

\lo{Tangent bundles in change action models.}

There is a natural, extrinsic, notion of higher-order derivative in change action models.
In such a model $\alpha : \cat{C} \to \ftor{\CAct}(\cat{C})$,
a $\cat{C}$-object $A$ is associated (via $\alpha$) with a change action, the carrier object of whose monoid is in turn associated with a change action, and so on \emph{ad infinitum}.
We construct a ``canonical'' change action model, $\ftor{\CAct_\omega}(\cat{C})$, that internalises such $\omega$-sequences that exhibit higher-order differentiation.
Objects of $\ftor{\CAct_\omega}(\cat{C})$ are $\omega$-sequences of ``contiguously compatible'' change actions; 
and morphisms are corresponding $\omega$-sequences of differential maps, each map being the canonical (via $\alpha$) derivative of the preceding in the $\omega$-sequence.
We show that $\ftor{\CAct_\omega}(\cat{C})$ is the final $\ftor{\CAct}$-coalgebra (relativised to change action models on $\cat{C}$).
The category $\ftor{\CAct_\omega}(\cat{C})$ may be viewed as a kind of Fa\`a di Bruno construction \cite{cruttwell2017cartesian,CockettS11} in the more general setting of change action models.

Change action models capture many versions of differentiation that arise in mathematics and computer science.
We illustrate their generality via three examples.
The first, \emph{(generalised) cartesian differential categories} (GCDC) \cite{blute2009cartesian,cruttwell2017cartesian}, are themselves an axiomatisation of the essential properties of the derivative.
We show that a GCDC $\cat{C}$---which by definition associates every object $A$ with a monoid $L(A) = (L_0(A), +_A, 0_A)$---gives rise to change action models in various non-trivial ways.
First there is a canonical change action model mapping each object $A$ to the trivial action of $L(A)$ on itself. 
A second arises from interpreting the map that zeroes the second component y $\Id : A \times L_0(A) \ra A \times L_0(A)$ as defining an action of $L(A)$ on $A$ in the Kleisli category of the tangent bundle monad.


Secondly we show how discrete differentiation in both the \emph{calculus of finite differences} \cite{jordan1965calculus} and \emph{Boolean differential calculus} \cite{steinbach2017boolean,thayse1981boolean} can be modelled using the full subcategory $\CGrpS$ of $\cat{Set}$ whose objects are groups.
Our unifying formulation generalises these discrete calculi to arbitrary groups, and gives an account of
the chain rule in these settings.

Our third example is differentiation of regular expressions.
Recall that Kleene algebra $\KK$ 
is the algebra of regular expressions.
Thanks to Taylor's Theorem \cite{DBLP:conf/lics/HopkinsK99}, every polynomial   over a commutative $\KK$, viewed as an endofunction on $\KK$ \emph{qua} (monoidal) change action, has a regular derivative.
We show that the algebra of polynomials over a commutative Kleene algebra is a change action model. 
Interestingly the derivatives are not additive in the second (i.e.~vectorial) argument, thus violating GCDC axiom [CD.2].

\paragraph{Outline.} In Section~\ref{sec:change-actions} we present the basic definitions of change actions and differential maps, and show how they can be organised into categories.
The theory of change action is extended to arbitrary cartesian categories $\cat{C}$ in Section~\ref{sec:change-actions-arbitrary}: we introduce the category $\ftor{\CAct}(\cat{C})$ of internal change actions on $\cat{C}$.
In Section~\ref{sec:extrinsic} we present change action models, and properties of the tangent bundle functors.
In Section~\ref{sec:examples} we illustrate the unifying power of change action models via three examples.
In Section~\ref{sec:omega-change-actions}, we study the category $\ftor{\CAct_\omega}(\cat{C})$ of $\omega$-change actions and $\omega$-differential maps.
Missing proofs are provided in the Appendix.

\section{Change actions}
\label{sec:change-actions}

  A \emph{change action} is a tuple $\da{A} = (\us{A}, \changes{A}, \oplus_A, +_A, 0_A)$ where $\us{A}$ and $\changes{A}$ are sets,
  $(\changes{A}, +_A, 0_A)$ is a monoid, and $\oplus_A : \us{A} \times \changes{A} \ra \us{A}$ is an action of the monoid on $\us{A}$.
We omit the subscript from $\oplus_A, +_A$ and $0_A$ whenever we can.

\begin{rem}
  Change actions are closely related to the notion of \emph{change structures} introduced in \cite{cai2014changes} but differ from the latter in not being dependently typed or assuming the existence of an $\ominus$ operator.
  On the other hand, change actions require the change set $\changes{A}$ to have a monoid structure compatible with the map $\oplus$, hence neither notion is strictly a generalisation of the other. Whenever one has a change structure, however, one can easily obtain a change action by considering the free monoid generated by its change set.
\end{rem}

\begin{defi}[Derivative condition]\rm
  Let $\da{A}$ and $\da{B}$ be change actions.
  A function $f : \us{A} \ra \us{B}$ is \emph{differentiable} if there is a function $\d f : \us{A} \times \changes{A} \ra \changes{B}$ satisfying
  \(
    f(a \oplus_A \change{a}) = f(a) \oplus_B \d f (a, \change{a}),
  \)
  for all $a \in \us{A}, \change{a} \in \changes{A}$.
  We call $\d f$ a \emph{derivative} for $f$, and write $f : \da{A} \ra \da{B}$ whenever $f$ is differentiable.
\end{defi}

\begin{lem}[Chain rule]
Given $f : \da{A} \ra \da{B}$ and $g : \da{B} \ra \da{C}$ with derivatives $\d f$ and $\d g$ respectively,
the function $\d (g \circ f) : \us{A} \times \changes{A} \to \changes{C}$ defined by
  \(
    \d(g \circ f)(a, \change a) \defeq \d g(f(a), \d f(a, \change a))
  \)
is a derivative for $g \circ f : \us A \to \us C$.
\end{lem}
\begin{proof}
Unpacking the definition, we have
$(g \circ f)(a) \oplus_C \d (g \circ f)(a, \change a)
= g(f(a)) \oplus_C \d g(f(a), \d f(a, \change a)) = g(f(a) \oplus_B \d f(a, \change a)) = g(f(a \oplus_A \change a))$, as desired.\qed
\end{proof}

\begin{exa}[Some useful change actions]
  \begin{enumerate}
  \item If $(A, +, 0)$ is a monoid, $(A, A, +, +, 0)$ is a change action (called \emph{monoidal}).
  \item 
    For any set $A$, $\da{A}_{\star} \defeq (A, \{\star\}, \pi_1, \pi_1, \star)$ is a (trivial) change action.
  \item Let $A \Ra B$ be the set of functions from $A$ from $B$, and $\arr{ev}_{A, B} : A \times (A \Ra B) \ra B$ be the usual evaluation map. 
  Then $(A, A \Ra A, \arr{ev}_{A,A}, \circ, \Id_A)$ is a change action.
  If $U \subseteq (A \Ra A)$ contains the identity map and is closed under composition, $(A, U, \arr{ev}_{A, A} \restriction_{A \times U}, \circ \restriction_{U \times U}, \Id_U)$ is a change action.
  \end{enumerate}
\end{exa}


\subsection{Regular derivatives}

The preceding definitions neither assume nor guarantee a derivative to be additive (i.e. they may not satisfy $\d f(x, \changes{a} + \changes{b}) = \d f(x, \changes{a}) + \d f(x, \changes{b})$), as they are in standard
differential calculus.
A strictly weaker condition that we will now require is \emph{regularity}: 
if a derivative is additive in its second argument then it is regular, but not vice versa. 
Under some condition, the converse is also true.

\begin{defi}\rm\label{def:regular}
  Given a differentiable map $f : \da{A} \ra \da{B}$, a derivative $\d f$ for $f$ is
  \emph{regular} if, for all $a \in \us{A}$ and $\change{a}, \change{b} \in \changes{A}$, we have $f(a, 0_A) = 0_B$ and $\d f(a, \change{a} +_A \change{b}) = \d f(a, \change{a}) +_B \d f(a \oplus_A \change{a}, \change{b})$.
\end{defi}

\def\lasttheorem{hello}
\begin{rep-proposition}{unique-derivative-regular}
  Whenever $f : \da{A} \ra \da{B}$ is differentiable and has a unique derivative $\d f$,
  this derivative is regular.
\end{rep-proposition}
\proofref
\begin{rep-proposition}{chain-rule-regular}
  Given $f : \da{A} \ra \da{B}$ and $g : \da{B} \ra \da{C}$ with regular derivatives $\d f$ and $\d g$ respectively, 
  the derivative $\d (g \circ f) = \d g \circ \pair{f \circ \pi_1}{\d f}$ is regular.
\end{rep-proposition}
\proofref
\subsection{Two categories of change actions}

The study of change actions can be undertaken in two ways: one can consider
functions that are differentiable (without choosing a derivative);
alternatively, the derivative itself can be considered part of the morphism.
The former leads to the category $\cat{\CAct^-}$, whose objects are change actions and morphisms are the differentiable maps.

The category $\cat{\CAct^-}$ was the category we originally proposed 
\cite{mario2019fixing}. 
It is well-behaved, possessing limits, colimits, and exponentials, which is a trivial corollary of the following result:

\begin{rep-theorem}{cact-preord}
  The category $\cat{\CAct^-}$ of change actions and differentiable morphisms is
  equivalent to $\cat{PreOrd}$, the category of preorders and monotone maps.
\end{rep-theorem}
\proofref

The actual structure of the limits and colimits in $\cat{\CAct^-}$ is, however, not so satisfactory. 
One can, for example, obtain the product of two change actions $\da{A}$ and $\da{B}$ by taking their product in $\cat{PreOrd}$ and turning it into a change action, 
but the corresponding monoid action map $\oplus$ is not, in general, easily expressible, even if those for $\da{A}$ and $\da{B}$ are. Derivatives of morphisms in $\cat{\CAct^-}$ can also be hard to obtain, as
exhibiting $f$ as a morphism in $\cat{\CAct^-}$ merely proves it is differentiable
but gives no clue as to how a derivative might be constructed. 

A more constructive approach is to consider morphism as a function together with a choice of a derivative for it.

\begin{defi}\rm
  Given change actions $\da{A}$ and $\da{B}$, a \emph{differential map} $\da{f} :
  \da{A} \ra \da{B}$ is a pair $(\us{f}, \d f)$ where $\us{f} : \us{A} \ra \us{B}$ is a function, and $\d f : \us{A} \times \changes A \ra \changes B$ is a regular derivative for $\us{f}$.
\end{defi}

The category $\cat{\CAct}$ has change actions as objects and differential maps as morphisms. 
  The identity morphisms are $(\Id_A, \pi_1)$; given morphisms $f : A \to B$ and $g : B \to C$, define the composite $\da{g} \circ \da{f} \defeq (\us{g} \circ \us{f}, \d g \circ \pair{\us{f} \circ \pi_1}{\d f}) : A \to C$.

Finite products and coproducts exist in $\cat{\CAct}$ (see Theorems~\ref{thm:cact-products} and~\ref{thm:cact-coproducts} for a more general statement).
Whether limits and colimits exist in $\cat{\CAct}$ beyond products and coproducts is open.  
\lo{It may be worth thinking along the lines of the example in Remark~\ref{rem:reg-der-not-nec}.}

\lo{What is the relationship between the two categories, $\cat{\CAct^-}$ and $\cat{\CAct}$?}
\lo{I presume $\cat{\CAct}$ is not cartesian closed?}
\todo{Figure these out.}

\begin{rem}
If one thinks of changes (i.e.~elements of $\changes A$) as morphisms between elements of $\us A$, then regularity resembles functoriality. 
This intuition is explored in Appendix~\ref{sec:2-categories}, where we show that categories of change actions organise themselves into 2-categories.
\end{rem}

\subsection{Adjunctions with $\cat{Set}$}

There is an obvious forgetful functor $\mathcal{F} : \cat{\CAct} \ra \cat{Set}$ that maps every change action $\da{A}$ to its underlying set $\us{A}$ and every differential map $\da{f} : \da{A} \ra \da{B}$ to the function on the underlying sets $\us f : \us A \ra \us B$.

Given a change action $\da{A}$ on a set $\us{A}$, the structure of the change action defines a preorder $\leq$ on $\us{A}$ where $a \leq b$ whenever there exists some $\change a$ such that $a \oplus \change a = b$ 
(indeed, one can think of change actions as particular presentations of preorders). 
Then we can define a quotient functor $\mathcal{Q} : \cat{\CAct} \ra \cat{Set}$ that maps the change action $\da{A}$ to the set $\us{A} /{\sim}$, 
where $\sim$ is the transitive  and symmetric closure of $\leq$; and the action on morphisms  $\da{f} : \da{A} \ra \da{B}$ is defined as $\mathcal{Q}(\da{f})(\left[ a \right]) = \left[ \us{f}(a) \right]$, 
where $[a]$ is the $\sim$-equivalence class of $a$.

Finally, there is a functor $\mathcal{D} : \cat{Set} \ra \cat{\CAct}$ that maps
every set $A$ to the discrete change action $(A, 1, \A{Id}, !, !)$ (where $1$
denotes the terminal object in $\cat{Set}$ and $!$ the universal morphism). This
functor $\mathcal{D}$ sends every function $f$ to the differential map $(f, !)$.

In what follows we will make use of the fact that
$\mathcal{F} \circ \mathcal{D}
= \mathcal{Q} \circ \mathcal{D}
= \ftor{Id}_{\cat{Set}}$.

\begin{rep-lemma}{forgetful-functor-is-right-adjoint}
  The forgetful functor $\mathcal{F}$ is right-adjoint to the functor
  $\mathcal{D}$, with the unit and counit given by:
  \begin{gather*}
    \epsilon : \mathcal{D} \circ \mathcal{F} \ra \ftor{Id}_{\cat{\CAct}}\\
    \epsilon_{\da{A}} = (\A{Id}_A, 0)\\
    \eta : \ftor{Id}_{\cat{Set}} \ra \mathcal{F} \circ \mathcal{D}
    = \ftor{Id}_{\cat{Set}}\\
    \eta_A = \A{Id}_A
  \end{gather*}
\end{rep-lemma}
\proofref

\begin{lem}
  The functor $\mathcal{D}$ is right adjoint to the quotient functor
  $\mathcal{Q}$, with unit and counit given by:
  \begin{gather*}
    \epsilon : \ftor{Id}_{\cat{Set}} \cong \mathcal{Q} \circ \mathcal{D}
    \ra \ftor{Id}_{\cat{Set}}\\
    \epsilon_A = \A{Id}_A\\
    \eta : \ftor{Id}_{\cat{\CAct}} \ra \mathcal{D} \circ \mathcal{Q}\\
    \eta_{\da{A}} = (\left[ \A{Id}_{\da{A}} \right], !)
  \end{gather*}
  where $\left[ \A{Id}_{\da{A}} \right]$ is the map that sends an element $a$ in
  $A$ to the equivalence class $\left[ a \right]$ of $a$ modulo $\sim_\leq$.
  Note that $\eta$ is well-defined since whenever $a \oplus \change a = b$ it
  is the case that $\left[ a \right] = \left[ b \right]$.
\end{lem}

In what follows we assume the Axiom of Choice. 
This is equivalent to the assumption that every set is the underlying set of some group. 
We will suppose a map $\mathcal{G}$ that sends each set $A$ to a group $\mathcal{G}_A$ whose underlying set is $A$.
\begin{defi}
  The functor $\mathcal{G} : \cat{Set} \ra \cat{\CAct}$ maps every
  object to the monoidal change action $(\mathcal{G}_A, \mathcal{G}_A, +, +,
  0)$, and every function $f : A \ra B$ to the differential map $(f, \d^1 f)$,
  with $\d^1f (x, \change x) = -f(x) + f(x + \change x)$. A straightforward
  consequence of this definition is that $\mathcal{G}$ is full and faithful.
\end{defi}

\begin{lem}
  The functor $\mathcal{G}$ is a right adjoint to the forgetful
  functor $\mathcal{F}$.
\end{lem}
\begin{proof}
  If one uses the hom-set isomorphism definition of adjunction, it follows
  trivially from the fact that every function into a change action of the form
  $\mathcal{G}(A)$ has one and only one derivative.
\end{proof}

\begin{rem}
  In a nutshell, this means there is a sequence of four adjunctions
  \begin{gather*}
    \mathcal{Q} \ts \mathcal{D} \ts \mathcal{F} \ts \mathcal{G}
  \end{gather*}
  where $\mathcal{Q}$ preserves finite products and $\mathcal{D}, \mathcal{G}$ are full and faithful.
\lo{What are the mathematical consequences of (non-topos) differential cohesion?}
  This is precisely the setting of Lawvere's notion of differential cohesion
  \cite{nLabDifferentialCohesive}
  (with the exception that $\cat{\CAct}$ is not a topos), which has been proposed to
  unify many settings for higher differential geometry.
  
  It is a topic of ongoing research to put this fact in the context of the recent
  advances in differential cohesive
  type theory \cite{grossdifferential,shulman2018brouwer}
  (the internal language of a topos with differential cohesion), which
  have recently been used to give a constructive formalization of Brouwer's fixed
  point theorem \cite{shulman2018brouwer}.
\end{rem}

\section{Change actions on arbitrary categories}
\label{sec:change-actions-arbitrary}
The definition of change actions makes no use of any properties of $\cat{Set}$ beyond the existence of products. 
Indeed, change actions can be characterised as just a kind of multi-sorted algebra, which is definable in any category with products.

\subsection{The category $\ftor{\CAct}(\cat{C})$}

Consider the category $\cart$ of (small) cartesian categories 
(i.e.~categories with chosen finite products) and product-preserving functors. 
We can define an endofunctor $\ftor{\CAct} : \cart \rightarrow \cart$ sending a category $\cat{C}$ to the category of (internal) change actions on $\cat{C}$.

The objects of $\ftor{\CAct}(\cat{C})$ are tuples 
$\da{A} = (\us A, \changes{A},\oplus_A, \plus_A, 0_A)$ where $\us A$ and $\changes{A}$ are (arbitrary) objects in $\cat{C}$,
$(\changes{A}, \plus_A, 0_A)$ is a monoid object in $\cat{C}$, 
and $\oplus_A : \us A \times \changes{A} \rightarrow \us A$ is a $\cat{C}$-morphism
such that the following diagrams---specifying monoid action---commute (omitting the obvious structural morphisms):
    \begin{center}
    \begin{tikzcd}
      \us A \arrow[rr, "\pair{\Id}{0_A}"] \arrow[rrd,swap, "\Id"]
      && \us A \times \changes{A} \arrow[d, "\oplus_A"]\\
      && \us A
    \end{tikzcd}\qquad
    \begin{tikzcd}
      \us A \times \changes{A} \times \changes{A}
      \arrow[d, swap, "\pair{\Id}{\plus_A}"] 
      \arrow[r, "\pair{\oplus_A}{\Id}"]
      & \us A \times \changes{A}
      \arrow[d, "\oplus_A"]\\
      \us A \times \changes{A}
      \arrow[r, swap, "\oplus_A"] & \us A
    \end{tikzcd}
    \end{center}

Given objects $\da{A}, \da{B}$ in $\ftor{\CAct}(\cat{C})$, 
the morphisms of $\ftor{\CAct}(\da{A}, \da{B})$ are pairs $\da{f} = (\us{f}, \d f)$ where
$\us f : \us A \ra \us B$ and $\d f : \us A \times \changes{A} \ra \changes{B}$ are morphisms in $\cat{C}$, such that the following diagrams commute:
    \begin{center}
      \begin{tikzcd}
        \us A \times \changes{A}
        \arrow[rr, "\pair{\us f \circ \pi_1}{\d f}"]
        \arrow[d, swap, "\oplus_A"] 
        && \us B \times \changes{B} \arrow[d, "\oplus_B"]\\
        \us A \arrow[rr, swap, "\us{f}"] && \us B
      \end{tikzcd}
      \qquad
      \begin{tikzcd}
        \us A \arrow[d, swap, "\pair{\Id}{0_A \circ !}"] \arrow[r, "!"] & 1 \arrow[d, "0_B"]\\
        \us A \times \changes{A} \arrow[r, swap, "\d f"] & \changes{B}
      \end{tikzcd}
    \end{center}
    \begin{center}
      \begin{tikzcd}
        \us A \times (\changes{A} \times \changes{A})
        \arrow[rr, "\pair{\pair{\pi_1}{\pi_1 \circ \pi_2}}{\A{a}}"]
        \arrow[dd, "\d f \circ (\Id \times +_A)", swap]
        && (\us A \times \changes{A}) \times ((\us A \times \changes{A}) \times \changes{A})
        \arrow[dd, "\d f \times (\d f \circ (\oplus_A \times \Id))"]
        \\\\
        \changes{B}
        &&
        \changes{B} \times \changes{B}
        \arrow[ll, "+_B"]
      \end{tikzcd}
    \end{center}
The first diagram states the derivative condition: $\us{f}(x \oplus_A \change x) = \us{f}(x) \oplus_B \d f(x, 
\change x)$. The other two assert a diagrammatic version of the regularity of
$\d f$.

The chain rule can then be expressed naturally by pasting two instances of the previous
diagram:
\begin{center}
  \begin{tikzcd}
    \us{A} \times \changes{A}
    \arrow[d, swap, "\oplus_A"]
    \arrow[rr, "\pair{\us{f} \circ \pi_1}{\d f}"] 
    \arrow[rrrr, bend left,
    "\pair{(\us{g} \circ \us{f}) \circ \pi_1}{\d g \circ \pair{\us{f} \circ \pi_1}{\d f}}"]
    &&
    \us B \times \changes{B} \arrow[d, "\oplus_B"]
    \arrow[rr, "\pair{\us{g} \circ \pi_1}{\d \us{g}}"] &&
    \us C \times \changes{C}
    \arrow[d, "\oplus_C"] \\
    \us A \arrow[rr, "\us{f}"]
    \arrow[rrrr, bend right, swap, "\us{g} \circ \us{f}"]
    && \us B
    \arrow[rr, "\us{g}"] && \us C
  \end{tikzcd}
\end{center}
Hence
\(
\da{f} \circ \da{g} = \pair{(\us g \circ \us f) \circ \pi_1}{\d g \circ \pair{\us f \circ \pi_1}{\d f}}.
\)

Now, given a product-preserving functor $\ftor{F} : \cat{C} \ra \cat{D}$, there
is a corresponding functor
$\ftor{\CAct}(\ftor{F}) : \ftor{\CAct}(\cat{C}) \ra \ftor{\CAct}(\cat{D})$
given by:
\begin{align*}
  \ftor{\CAct}(\ftor{F})(\us{A}, \changes{A}, \oplus_A, \plus_A, 0_A) 
  &\defeq (\ftor{F}(\us{A}), \ftor{F}(\changes{A}), \ftor{F}(\oplus_A), \ftor{F}(\plus_A), \ftor{F}(0_A))\\
  \ftor{\CAct}(\ftor{{F}})(\us f, \d f) &\defeq (\ftor{F}(\us {f}), \ftor{F}(\d f))
\end{align*}

We can embed $\cat{C}$ fully and faithfully into $\ftor{\CAct}(\cat{C})$ via the functor
$\eta_{\cat{C}}$ which sends an object $A$
of $\cat{C}$ to the ``trivial'' change action $\da{A}_\star = (A, \top, \pi_1, !,
!)$ and every morphism $f : A \ra B$ of $\cat{C}$ to the morphism $(f, !)$. As
before, this functor extends to a natural transformation from the identity
functor to $\ftor{\CAct}$.

Additionally, there is an obvious forgetful functor $\epsilon_{\cat{C}} : \ftor{\CAct}(\cat{C}) \ra \cat{C}$,
which defines the components of a natural transformation $\epsilon$ from the functor $\ftor{\CAct}$ to the identity endofunctor $\Id$.

Given $\cat{C}$, we write $\xi_{\cat{C}}$ for the functor $\ftor{\CAct}(\epsilon_{\cat{C}}) : \ftor{\CAct}(\ftor{\CAct}(\cat{C})) \ra \ftor{\CAct}(\cat{C})$. \footnote{
  One might expect $\ftor{\CAct}$ to be a comonad with $\epsilon$ as a counit. But if this were the case, we would have $\xi_{\cat{C}} = \epsilon_{\ftor{\CAct}(\cat{C})}$, which
  is, in general, not true.
} Explicitly, this functor maps an object $(\da{A}, \da{B}, \da{\oplus}, \da{\plus}, \da{0})$ in $\ftor{\CAct}(\ftor{\CAct}(\cat{C}))$ to
the object $(\us{A}, \us{B}, \us\oplus, \us\plus, \us 0)$.
Intuitively, $\epsilon_{\ftor{\CAct}(\cat{C})}$ prefers the ``original'' structure on objects, whereas
$\xi_{\cat{C}}$ prefers the ``higher'' structure. The equaliser of these two
functors is precisely the category of change actions whose higher structure is
the original structure.

\subsection{Products and coproducts in $\ftor{\CAct}(\cat{C})$}

We have defined $\ftor{\CAct}$ as an endofunctor on cartesian categories.
This is well-defined: if $\cat{C}$ has all finite (co)products, so does $\ftor{\CAct}(\cat{C})$.
Let $A = (\us A, \changes A, \oplus_A, +_A, 0_A)$ and $B = (\us B, \changes B, \oplus_B, +_B, 0_B)$ be change actions on $\cat{C}$.
We present their product and coproducts as follows. 

\begin{rep-theorem}{cact-products}
  The following change action is the product of $A$ and $B$ in $\ftor{\CAct}(\cat{C})$
  \begin{gather*}
    A \times B \defeq (\us A \times \us B, \changes A \times \changes B, \oplus_{A \times B}, +_{A \times B}, \pair{0_A}{0_B})
  \end{gather*}
where
$\oplus_{A \times B} \defeq \pair{\oplus_A \circ (\pi_1 \times \pi_1)}{\oplus_B \circ (\pi_2 \times \pi_2)}$ and $+_{A \times B} \defeq \langle+_A \circ (\pi_1 \times \pi_1), +_B \circ (\pi_2 \times \pi_2)\rangle$.
The projections are
\( \overline{\pi_1} = (\pi_1, \pi_1 \circ \pi_2)
\)
and
\(
\overline{\pi_2} = (\pi_2, \pi_2 \circ \pi_2),\)
writing $\overline{f}$ for maps $f$ in $\ftor{\CAct}$ to distinguish them from $\cat{C}$-maps.
\end{rep-theorem}
\proofref

\begin{rep-theorem}{cact-terminal}
  The change action $\overline{\top} = (\top, \top, \pi_1, \pi_1, \Id_\top)$ is the terminal object in $\ftor{\CAct}(\cat{C})$, where $\top$ is the terminal object of $\cat{C}$.
  Furthermore, if $A$ is a change action every point $\us{f} : \top \ra \us{A}$ in $\cat{C}$ is differentiable, with (unique) derivative $0_A$.
\end{rep-theorem}
\proofref

Whenever we have a differential map $f : A \times B \ra C$ between change actions, we can compute its derivative $\d f$ by adding together its ``partial'' derivatives:\footnote{
  Alternatively, one can define the (first) partial derivative of $f(x, y)$ as a map $\delta_1 f$ such that $f(x \oplus \change x, y) = f(x, y) \oplus \delta_1(x, y, \change x)$.
  It can be shown that a map is differentiable iff its first and second derivatives exist.
}.
\begin{rep-lemma}{partial-derivatives}
    Let $f : A \times B \ra C$ be a differential map.
    Then \[\d f((a, b), (\change a, \change b)) = +_C \circ \pair{\d f((a, b), (\change a, 0_B))}{\d f((\oplus_A \circ \pair{a}{\change a}, b), (0_A, \change b))} \]
    (The notational abuse is justified by the internal logic of a cartesian category.)
\end{rep-lemma}
\proofref

\begin{rep-theorem}{cact-coproducts}
  If $\cat{C}$ is distributive, with law $\delta_{A, B, C} : (A \sqcup B) \times C \ra (A \times C) \sqcup (B \times C)$, the following change action is the coproduct of $A$ and $B$ in $\ftor{\CAct}(\cat{C})$
  \begin{gather*}
    A \sqcup B \defeq (\us A \sqcup \us B, \changes A \times \changes B, \oplus_{A \sqcup B}, +_{A \sqcup B}, \pair{0_A}{0_B})
  \end{gather*}
where 
$\oplus_{A \sqcup B} \defeq \left[ \oplus_A \circ (\Id_A \times \pi_1), \oplus_B \circ (\Id_B \times \pi_2) \right] \circ \delta_{A, B, C}$, and $+_{A \sqcup B} \defeq \langle +_A \circ (\pi_1 \times \pi_1), +_B \circ (\pi_2 \times \pi_2)\rangle$.
The injections are
  \( \overline{\iota_1} = (\iota_1, \pair{\pi_2}{0_B})
  \) 
  and
  \(\overline{\iota_2} = (\iota_2, \pair{0_A}{\pi_2})
  \).  
\end{rep-theorem}
\proofref

\subsection{Change actions as Lawvere theories}

According to their definition, change actions seem nothing more than multi-sorted
algebras. This is perhaps misleading in that it suggests that differential
maps should correspond to algebra homomorphisms, which is in fact false:
a homomorphism of change actions $\da{A}, \da{B}$ would be a pair $(u, v)$ 
where $\us{u} : \us A \ra \us B$ is a function and $v : \changes{A} \ra \changes{B}$ is
a monoid homomorphism such that $u(a \oplus \change{a}) = u(a) \oplus v(\change{a})$.
That is to say, a homomorphism of change actions as algebras is precisely a
homomorphism of monoid actions.

There is a sense, however, in which differential maps are exactly algebra
homomorphisms. To make this precise, we require a few new definitions.

\begin{defi}
  The $\changes$-theory $\theo{\changes}$ is the free Cartesian category
  generated by objects $X, \changes{X}$ and morphisms
  \begin{gather*}
    \oplus : X \times \changes{X} \ra X\\
    + : \changes{X} \times \changes{X} \ra \changes{X}\\
    0 : 1 \ra \changes{X}
  \end{gather*}
  subject to the equations:
  \begin{gather*}
    \oplus \circ \pair{\A{Id}}{0 \circ !} = \A{Id}\\
    \oplus \circ (\A{Id} \times +)
    = \oplus \circ (\oplus \times \A{Id})\circ \alpha^{-1} 
  \end{gather*}
  A $\changes$-algebra on a Cartesian category $\cat{C}$ is a product-preserving
  functor from $\theo{\changes}$ into $\cat{C}$.
\end{defi}

\begin{rem}
  Every $\changes$-algebra $F : \theo{\changes} \ra \cat{C}$ corresponds to a
  change action on $\cat{C}$ given by:
  \begin{gather*}
    \da{F} = (F(X), F(\changes X), F(\oplus), F(+), F(0))
  \end{gather*}
  Conversely, every change action on $\cat{C}$ induces a $\changes$-algebra.
  However, $\changes$-algebra homomorphisms do not correspond to differentiable
  morphisms, hence the category of $\changes$-algebras and $\changes$-algebra
  homomorphisms is not equivalent to the category $\CAct$ of change actions.
\end{rem}

\begin{defi}
  The $\A{T}$-theory $\theo{\A{T}}$ is the free Cartesian category
  generated by objects $X, \A{T}X$ and morphisms
  \begin{gather*}
    \oplus : \A{T}X \ra X\\
    \Pi    : \A{T}X \ra X
  \end{gather*}
  A $\A{T}$-algebra on a Cartesian category $\cat{C}$ is a product-preserving
  functor from $\theo{\A{T}}$ into $\cat{C}$. A homomorphism of $\A{T}$-algebras
  $F, G$ is a natural transformation $\phi : F \ra G$.
\end{defi}

\begin{rep-lemma}{delta-algebras-are-t-algebras}
  Consider the product-preserving functor $\mathbb{T} : \theo{\A{T}} \ra
  \theo{\changes}$ defined by:
  \begin{gather*}
    \mathbb{T}(X) = X\\
    \mathbb{T}(\A{T}X) = X \times \changes{X}\\
    \mathbb{T}(\oplus) = \oplus\\
    \mathbb{T}(\Pi) = \pi_1
  \end{gather*}
  Every $\changes$-algebra $F$ corresponds then to a $\A{T}$-algebra
  $F \circ \mathbb{T}$. Furthermore, given $\changes$-algebras $F, G$, there is
  a one-to-one correspondence between $\mathbb{T}$-algebra homomorphisms
  $\phi : F \circ \mathbb{T} \ra G \circ \mathbb{T}$ and pairs $(f, f')$ of
  a differentiable function $f$ and its derivative $f'$ between the underlying
  change actions $\da{F}, \da{G}$.
\end{rep-lemma}
\proofref

These definitions exhibit $\changes$-algebras and $\A{T}$-algebras as
multi-sorted Lawvere algebras. Differential morphisms between
$\changes$-algebras are precisely Lawvere homomorphisms when the corresponding
$\changes$-algebras are regarded as $\A{T}$-algebras.

The parallel between $\changes$-algebras and $\A{T}$-algebras is strikingly
similar to the connection between (generalized) Cartesian differential categories
and categories with tangent structure that was outlined in \cite{cockett2014differential}.

\subsection{Stable derivatives and additivity}

We do not require derivatives to be additive in their second argument; indeed in many cases they are not. 
Under some simple conditions, however, (regular) derivatives can be shown to be additive.

\lo{In this subsection, we have used both ``linear'' and ``additive'' interchangeably.
Clarify this / make this consistent.}

\begin{defi}\rm
  Given a (internal) change action $\da{A}$ and arbitrary objects $\us B, \us C$ in a cartesian category $\cat{C}$,
  a morphism
  $u : \us A \times \us B \ra \us C$
  is \emph{stable} whenever the 
  following diagram commutes:
  \begin{center}
    \begin{tikzcd}
      (\us A \times \changes A) \times \us B
      \arrow[d, "\pi_1 \times \Id"]
      \arrow[r, "\oplus \times \Id"]
      &\us A \times \us B
      \arrow[d, "u"]
      \\
      \us A \times \us B
      \arrow[r, "u"]
      & \us C
    \end{tikzcd}
  \end{center}  
\end{defi}

\medskip

If one thinks of $\changes A$ as the object of ``infinitesimal'' transformations on $\us A$, then the preceding definition says that a morphism $u : \us A \times \us B \ra \us C$ is stable whenever infinitesimal changes on the input $A$ do not affect its output.

\begin{rep-lemma}{stable-derivatives-additive}
  Let $\da{f} = (\us f, \d f)$ be a differential map in $\ftor{\CAct}(\cat{C})$. If
  $\d f$ is stable, then it is additive in its second argument\footnote{
    Note that the converse is not the case, i.e. a derivative can be additive but not stable.
  }, i.e. the following
  diagram commutes:
  \begin{center}
    \begin{tikzcd}[ampersand replacement=\&]
      \us A \times (\changes A \times \changes A)
      \arrow[rrr, "\pair{\pair{\pi_1}{\pi_1 \circ \pi_2}}{\pair{\pi_1}{\pi_2\circ\pi_2}}"]
      \arrow[dd, "\d f \circ (\Id_{\us{A}} \times +_A)", swap]
      \& \& \& (\us A \times \changes A) \times (\us A \times \changes A)
      \arrow[dd, "\d f \times \d f"]
      \\
      \\
      \changes{B}
      \& \& \&
      \changes B \times \changes{B}
      \arrow[lll, "+_B"]
    \end{tikzcd}
  \end{center}
\end{rep-lemma}
\proofref
\begin{rep-lemma}{composition-stable}
  Let $\da{f} = (\us f, \d f)$ and $\da{g} = (\us g, \d g)$ be differential maps, with
  $\d g$ stable. Then $\d (\da{g} \circ \da{f})$ is stable.
\end{rep-lemma}
\proofref
It is straightforward to see that the category $\ftor{Stab}(\cat{C})$ of change actions and differential maps with stable derivatives is a subcategory of $\ftor{\CAct}(\cat{C})$.

\section{Higher-order derivatives: the extrinsic view}
\label{sec:extrinsic}

\lochanged{In this section we study categories in which every object is equipped with a change action, and every morphism specifies a corresponding differential map.
This provides a simple way of characterising categories which are models of higher-order differentiation purely in terms of change actions.}

\subsection{Change action models}

\lochanged{Recall that a \emph{copointed endofunctor} is a pair $(\ftor{F}, \sigma)$ where the endofunctor $\ftor{F} : \cat{C} \ra \cat{C}$ is equipped with a natural transformation $\sigma : \ftor{F} \xrightarrow{.} \ftor{Id}$.
A \emph{coalgebra of a copointed endofunctor $(\ftor{F}, \sigma)$} is an object $A$ of $\cat{C}$ together with a morphism $\alpha : A \ra \ftor{F}A$ such that $\sigma_A \circ \alpha = \Id_A$.}

\begin{defi}\rm
We call a coalgebra $\alpha : \cat{C} \ra \ftor{\CAct}(\cat{C})$ of the copointed endofunctor $(\ftor{\CAct}, \epsilon)$ a \emph{change action model} (on $\cat{C})$.
\end{defi}

\noindent\emph{Assumption}. Throughout Sec.~\ref{sec:extrinsic}, we fix a change action model $\alpha : \cat{C} \to \ftor{\CAct}(\cat{C})$.

Given an object $A$ of $\cat{C}$, the coalgebra $\alpha$ specifies a (internal) change action 
$\alpha(A) = (A, \changes A, \oplus_A, +_A, 0_A)$ in $\ftor{\CAct}(\cat{C})$. 
(We abuse notation and write $\changes A$ for the \lochanged{carrier object of the monoid specified in $\alpha(A)$}; similarly for $+_A, \oplus_A$, and $0_A$.) 
Given a morphism $f : A \ra B$ in $\cat{C}$,
there is an associated differential map $\alpha(f) = (f, \d f) : \alpha(A) \ra \alpha(B)$.
Since $\d f : A \times \changes{A} \ra \changes{B}$ is also a $\cat{C}$-morphism, there is a corresponding differential map $\alpha(\d f) = (\d f, \d^2 f)$ in $\ftor{\CAct}(\cat C)$, where $\d^2 f : (A \times \changes A) \times (\changes A \times \changes^2A) \ra \changes^2 {B}$ is a second derivative for $f$. 
Iterating this process, we obtain an $n$-th derivative $\d^n f$ for every $\cat{C}$-morphism $f$.
Thus change action models offer a setting for reasoning about higher-order differentiation.

\subsection{Tangent bundles in change action models}

In differential geometry the tangent bundle functor, which maps every manifold to its tangent bundle, is an important construction. 
There is an endofunctor on change action models reminiscent of the tangent bundle functor, with analogous properties.

\begin{defi}\rm
  The \emph{tangent bundle functor} $\ftor{T} : \cat{C} \ra \cat{C}$ is defined as $\ftor{T}A \defeq A \times \changes A$ and $\ftor{T}f \defeq \pair{f \circ \pi_1}{\d f}$.
\end{defi}

\noindent\emph{Notation}. We use shorthand $\pi_{ij} \defeq \pi_i \circ \pi_j$.

  The tangent bundle functor $\ftor{T}$ preserves products up to isomorphism,
  i.e.~for all objects $A, B$ of $\cat{C}$, we have $\ftor{T}(A \times B) \cong \ftor{T}A \times \ftor{T}B$ and $\ftor{T}1 \cong 1$. 
  In particular, 
    $\phi_{A, B} \defeq \pair{\pair{\pi_{11}}{\pi_{12}}}
    {\pair{\pi_{21}}{\pi_{22}}} : \ftor{T}A \times \ftor{T}B \ra \ftor{T}(A \times B)$
  is an isomorphism.
  Consequently, given maps $f : A \ra B$ and $g : A \ra C$, then, up to the previous
  isomorphism, $\ftor{T}\pair{f}{g} = \pair{\ftor{T}f}{\ftor{T}g}$.

A consequence of the structure of products in $\ftor{\CAct}(\cat{C})$ is that the map
$\oplus_{A \times B}$ inherits the pointwise structure in the following sense:
\begin{lem}
  Let $\phi_{A, B} : \ftor{T}A \times \ftor{T}B \ra \ftor{T}(A \times B)$ be the canonical isomorphism described above. Then $\oplus_{A \times B} \circ \phi_{A, B} = \oplus_A \times \oplus_B$.
\end{lem}

It will often be convenient to operate directly on the functor $\ftor{T}$, rather than on the underlying derivatives. 
For these, the following results are useful:
\begin{rep-lemma}{tangent-bundle-monad}
  The following families of morphisms are natural transformations:
  \(
  \pi_1, \oplus_A : \ftor{T}(A) \ra A, \
    \arr{z} \defeq \pair{\Id}{0} : A \ra \ftor{T}(A), \
    \arr{l} \defeq \pair{\pair{\pi_1}{0}}{\pair{\pi_2}{0}} : \ftor{T}(A) \ra \ftor{T}^2(A).
  \)  
  Additionally, the triple $(\ftor{T}, \arr{z}, \ftor{T}\oplus)$ defines a monad on $\cat{C}$.
\end{rep-lemma}
\proofref
\todo{
  More lemmas about the tangent bundle functor? Hopefully something about
  vector fields, maybe even Leibniz's rule.
}

A particularly interesting class of change action models are those that are also cartesian closed.
Surprisingly, this has as an immediate consequence that differentiation is itself internal to the category.

\begin{rep-lemma}{internalisation}[Internalisation of derivatives]
  Whenever $\cat{C}$ is cartesian closed, there is a morphism $\arr{d}_{A, B} : (A \Ra B) \ra (A \times \changes{A}) \Ra \changes{B}$ such that, 
  for any morphism $f : 1 \times A \ra B$, $\arr{d}_{A, B} \circ \Lambda f = \Lambda (\d f \circ \pair{\pair{\pi_1}{\pi_{12}}}{\pair{\pi_1}{\pi_{22}}})$.
\end{rep-lemma}
\proofref

Under some conditions, we can classify the structure of the exponentials in $(\ftor{\CAct}, \epsilon)$-coalgebras. 
This requires the existence of an infinitesimal object.\footnote{The concept  ``infinitesimal object'' is borrowed from synthetic differential geometry \cite{Kock06}. 
However, there is nothing intrinsically ``infinitesimal'' about these objects here.}

\begin{defi}\rm
  If $\cat{C}$ is cartesian closed,
  an \emph{infinitesimal object} $D$ is an object of $\cat{C}$ such that the
  tangent bundle functor $\ftor{T}$ is represented by the covariant Hom-functor
  $D \Ra (\cdot)$, i.e. there is a natural isomorphism $\phi : (D \Ra (\cdot)) \xrightarrow{.} \ftor{T}$.
\end{defi}
\begin{lem}
  Whenever there is an infinitesimal object in $\cat{C}$, the tangent bundle
  $\ftor{T}(A \Ra B)$ is naturally isomorphic to $A \Ra \ftor{T}B$.
\end{lem}

We would like the tangent bundle functor to preserve the exponential
structure; in particular we would expect a result of the form $\der{(\lambda y . t)}{x} = \lambda y . \der{t}{x}$,
which is true in differential $\lambda$-calculus \cite{DBLP:journals/tcs/EhrhardR03}. 
Unfortunately it seems impossible to prove in general that this equation holds, although weaker results are available.
If the tangent bundle functor is representable, however, additional structure is preserved.

\begin{rep-theorem}{oplus-lambda-infinitesimal}
  The isomorphism between the functors $\ftor{T}(A \Ra (\cdot))$ and $A \Ra \ftor{T}(\cdot)$ respects the structure of $\ftor{T}$, in the sense that the following diagram commutes:
  \end{rep-theorem}
\begin{center}
    \begin{tikzcd}[ampersand replacement=\&]
      \ftor{T}(A \Ra B)
      \arrow[d, "\oplus_{A \Ra B}", swap]
      \arrow[r, "\cong"]
      \& A \Ra \ftor{T}(B)
      \arrow[ld, "\Id_A \Ra \oplus_B"]
      \\
      A \Ra B
    \end{tikzcd}
\end{center}
\proofref
\todo{
  Can one differentiate under a lambda in these conditions?
}


\section{Examples of change action models}

\label{sec:examples}

\subsection{Generalised cartesian differential categories}

\emph{Generalised cartesian differential categories} (GCDC) \cite{cruttwell2017cartesian}---a recent generalisation of cartesian differential categories \cite{blute2009cartesian}---are models of differential calculi.
We show that change action models generalise GCDC in that GCDCs give rise to change action models in \lochanged{three}\footnote{\lochanged{The third, the Eilenberg-Moore model, is presented in Appendix~\ref{sec:em-model}}.} different (non-trivial) ways.
In this subsection let $\cat{C}$ be a GCDC (we assume familiarity with the definitions and notations in \cite{cruttwell2017cartesian}).


\paragraph{1.~The Flat Model.}
\label{sec:flat-model}
Define the functor $\alpha : \cat{C} \ra \ftor{\CAct}(\cat{C})$ as follows.
Let $f : A \ra B$ be a $\cat{C}$-morphism.
Then $\alpha(A) \defeq (A, L_0(A), \pi_1, +_A, 0_A)$ and $\alpha(f) \defeq (f, \DF{f})$.

\begin{rep-theorem}{GCDC-change-action-model}
  The functor $\alpha$ is a change action model.
\end{rep-theorem}

\proofref

\paragraph{2.~The Kleisli Model.} GCDCs admit a tangent bundle functor, defined analogously to the standard notion in differential geometry. 
Let $f : A \to B$ be a $\cat{C}$-morphism.
Define the \emph{tangent bundle functor} $\ftor{T} : \cat{C} \ra \cat{C}$ as: $\ftor{T}A \defeq A \times L_0(A)$, and $\ftor{T}f \defeq \pair{f \circ \pi_1}{\DF{f}}$.
The functor $\ftor{T}$ is in fact a monad, with unit $\eta = \pair{\Id}{0_A} : A \ra A \times L_0(A)$ and multiplication $\mu : (A \times L_0(A)) \times L_0(A)^2 \ra A \times L_0(A)$ defined by the composite:
\[
(A \times L_0(A)) \times L_0(A)^2 
      \xrightarrow{\pair{\pi_1 \circ \pi_1}{\pair{\pi_2 \circ \pi_1}{\pi_1 \circ \pi_2}}}
      A \times L_0(A)^2
      \xrightarrow{\Id \times +_A}
      A \times L_0(A)
\]
Thus we can define the Kleisli category of this functor by $\cat{C}_{\ftor{T}}$ which has geometric significance as a category of generalised vector fields. 

We define the functor $\alpha_{\ftor{T}} : \cat{C}_{\ftor{T}} \ra \ftor{\CAct}(\cat{C}_{\ftor{T}})$: 
given a $\cat{C}_{\ftor{T}}$-morphism $f : A \to B$, set
  \(
  \alpha_{\ftor{T}}(A) \defeq (A, L_0(A), \eta \circ \pi_1, \eta \circ +_A, \eta \circ 0_A)
  \) and
  \(
  \alpha_{\ftor{T}}(f) \defeq (f, \DF{f})
  \)\footnote{
    An earlier version of this work defined the action $\oplus_A$ as $\Id_{A} \times \Id_{L_0(A)}$. This
    is in fact incorrect as the resulting derivatives fail to satisfy regularity.
  }.
\begin{rep-lemma}{lem:GCDC-change-action-model}
$\alpha_{\ftor{T}}$ is a change action model.
\end{rep-lemma}
\proofref

\begin{rem}\label{rem:GCDC-axioms}
The converse is not true: in general the existence of a change action model on $\cat{C}$ does not imply that $\cat{C}$ satisfies the GCDC axioms. 
However, if one requires, additionally, $(\changes A, +_A, 0_A)$ to be commutative, with $\changes (\changes A) = \changes A$ and $\oplus_{\Delta A} = +_A$ for all objects $A$,
and some technical conditions (stability and uniqueness of derivatives), then it can be shown that $\cat{C}$ is indeed a GCDC.
\end{rem}

\map{Not sure whether to mention that we are convinced that these are sufficient but not necessary.}
\lo{No need. BTW which assumptions might not be necessary?}
\lo{SELF-NOTE.
- [CD.1]: differential compatible with sum. 
  
- [CD.2]: derivative additive in 2nd argument.
  
- [CD.3,4]: differential compatible with products.

- [CD.5]: chain rule

- [CD.6]: derivative linear in 2nd argument. 

- [CD.7]: symmetry of second partial derivatives.
}

\lo{General question. 
Every cartesian differential category gives rise to canonical change action models.
What can we say about the other direction?
Given a change action model, what additional conditions (in addition to ``strictness'') will guarantee satisfaction of [CD.1-7]?
We already know that if all the $\d f$ morphisms are stable then [CD.2] is satisfied.}
\todo{
  Conjecture:
}
\map{\begin{itemize}
  \item First of all, we need $\Delta \Delta A = \Delta A$ for the types to check.
  \item Second, we need all monoids to be commutative.
  \item For [CD.1]: the second part is true for free. The first part is implied by requiring
    $\oplus_{\Delta A} = +_A$ (plus commutativity and thinness), but it's not equivalent to it.
  \item For [CD.2]: stability implies CD.2 - not sure if it's equivalent.
  \item For [CD.3]: we get it for free
  \item For [CD.4]: we get it for free
  \item For [CD.5]: we get it for free
  \item For [CD.6]: Implied by thinness + $\oplus_{\Delta A} = +_A$.
  \item For [CD.7]: Implied by thinness + commutativity + $\oplus_{\Delta A} = +_A$
\end{itemize}
Mostly, I think the reasonable requirements are stability, commutativity and $\oplus_{\Delta A} = +_A$.
The only one that bothers me is thinness, since I feel like it should not be necessary - but it
seems to be. The issue boils down to: without thinness, there is ``a'' derivative with the required
property, but no guarantee that it's the one selected by the change action model.}
\lo{This is useful.}
\map{Another note: thinness + additivity implies stability, i.e. for derivatives between thin change actions additive=stable}
\todo{Actually add a section about thinness}

\subsection{Difference calculus and Boolean differential calculus}

Consider the full subcategory $\CGrpS$ of $\cat{Set}$ whose objects are all the groups\footnote{
  We consider arbitrary functions, rather than group homomorphisms, since, according to this change action structure, every function between groups is differentiable.
}.
This is a cartesian closed category which can be endowed with the structure of a $(\ftor{\CAct}, \epsilon)$-coalgebra $\alpha$ in a straightforward way.

  Given a group $A = (A, +, 0, -)$, define change action
  \(
    \alpha(A) \defeq  (A, A, +, +, 0).
  \)
  Given a function $f : A \ra B$,  define differential map $\alpha(f) \defeq  (f, \d f)$ where
  \(
  \d f(x, \change{x}) \defeq -f(x) + f(x \oplus \change{x}).
  \)
  Notice 
  \(
  f(x) \oplus \d f(x, \change x) 
  = f(x) + (-f(x) + f(x + \change{x}))
  = f(x + \change{x})
  = f(x \oplus \change{x});
  \)
  hence $\d f$ is a (regular) derivative\footnote{
    Note that $\d f$ need not be additive in its second argument, and so derivatives in $\CGrpS$ do not satisfy all the axioms of a cartesian differential category.
  } for $f$, and $\alpha(f)$ a map in $\ftor{\CAct}(\CGrpS)$. The following result is then immediate.

\begin{lem}
$\alpha : \CGrpS \to \ftor{\CAct}(\CGrpS)$ defines a change action model.
\end{lem}


This result is interesting. 
In the calculus of finite differences \cite{jordan1965calculus}, 
the \emph{discrete derivative} (or \emph{discrete difference operator})
of a function $f : \ZZ \ra \ZZ$ is defined as
\(
  \delta f(x) \defeq f(x + 1) - f(x)
\).
In fact the discrete derivative $\delta f$ is (an instance of) the derivative of $f$ \emph{qua} morphism in $\CGrpS$, i.e.~$\delta f(x) = \d f(x, 1)$.

Finite difference calculus \cite{gleich2005finite,jordan1965calculus} has found applications in combinatorics and numerical computation. 
Our formulation via change action model over $\CGrpS$ has several advantages. 
First it justifies the chain rule, which seems new.
Secondly, it generalises the calculus  to arbitrary groups.
To illustrate this, consider the \emph{Boolean differential calculus} \cite{steinbach2017boolean,thayse1981boolean},
a theory that applies methods from calculus to the space $\BB^n$ of vectors of elements of some Boolean algebra $\BB$.


\begin{defi}\rm
  Given a Boolean algebra $\BB$ and function $f : \BB^n \ra \BB^m$, the \emph{$i$-th Boolean derivative of $f$ at $(u_1, \ldots, u_n) \in \BB^n$} is the value
 \(   
    \frac{\d f}{\d x_i}(u_1, \ldots, u_n) \defeq 
    f(u_1, \ldots, u_n) \nlra f(u_1, \ldots, \neg u_i, \ldots, u_n),
\)
writing $u \nleftrightarrow v \defeq (u \wedge \neg v) \vee (\neg u \wedge v)$ for exclusive-or.
\end{defi}

Now $\BB^n$ is a $\CGrpS$-object.
Set 
$
  \top_i \defeq (\bot, \overset{i-1}{\ldots}, \bot, \top, \bot, \overset{n - i}{\ldots}, \bot) \in \BB^n
$.
\begin{rep-lemma}{lem:boolean-derivative}
The Boolean derivative of $f : \BB^n \ra \BB^m$ coincides with its
derivative qua morphism in $\CGrpS$: $\frac{\d f}{\d x_i}(u_1, \ldots, u_n) = \d f((u_1, \ldots, u_n), \top_i)$.
\end{rep-lemma}

\proofref


\subsection{Polynomials over commutative Kleene algebras}

The algebra of polynomials over a commutative Kleene algebra \cite{DBLP:conf/lics/HopkinsK99,Kleene56} (see \cite{DBLP:conf/latin/LombardyS04,DBLP:journals/jacm/EsparzaKL10} for work of a similar vein) is a change action model.
Recall that Kleene algebra is the algebra of regular expressions \cite{DBLP:journals/jacm/Brzozowski64,Conway71}.
Formally a \emph{Kleene algebra} $\KK$ is a tuple $(K, +, \cdot, {}^\star, 0, 1)$ such that
$(K, +, \cdot, 0, 1)$ is an idempotent semiring under $+$ satisfying, for all $a, b, c \in K$: 
\begin{align*}
    1 + a \, a^\star &= a^\star\quad
    1 + a^\star a = a^\star\quad
    b + a \, c \leq c \ra a^\star \, b \leq c
    \quad
    b + c \, a \leq c \ra b \, a^\star \leq c
\end{align*}
where ${a \leq b} \defeq a + b = b$.
A Kleene algebra is \emph{commutative} whenever $\cdot$ is.

Henceforth fix a commutative Kleene algebra $\KK$. 
Define the \emph{algebra of polynomials} $\KK[\overline{x}]$ as the free extension of the algebra $\KK$ with elements $\overline{x} = x_1, \ldots, x_n$. 
We write $p(\overline{a})$ for the value of $p(\overline{x})$ evaluated at $\overline x \mapsto \overline a$. 
Polynomials, viewed as functions, are closed under composition: 
when $p \in \KK[\overline{x}], q_1, \ldots, q_n \in \KK[\overline{y}]$ are  polynomials, so is the composite $p(q_1(\overline{y}), \ldots, q_n(\overline{y}))$.

Given a polynomial $p = p(\overline x)$, we define its \emph{$i$-th derivative} 
$\der{p}{x_i}(\overline x) \in \KK[\overline x]$:
\begin{align*}
  \der{a}{x_i} (\overline x) &= 0
  \qquad
  \der{p^\star}{x_i} (\overline x) = p^\star(\overline x) \, \der{p}{x_i} (\overline x)
  \qquad
  \der{x_j}{x_i} (\overline x) = 
  \begin{cases}
  1 \text{ if $i = j$}\\
  0 \text{ otherwise}
  \end{cases}
  \\
  \der{(p+q)}{x_i}(\overline x) &= \der{p}{x_i} (\overline x) + \der{q}{x_i} (\overline x)\qquad
  \der{(p \, q)}{x_i} (\overline x) = p(\overline x) \der{q}{x_i} (\overline x)+ q (\overline x)\der{p}{x_i}(\overline x)
\end{align*}
Write $\der{p}{x_i}(\overline e)$ to mean the result of evaluating the polynomial $\der{p}{x_i}(\overline x)$ at $\overline x \mapsto \overline e$.
\begin{thm}[Taylor's formula \cite{DBLP:conf/lics/HopkinsK99}]\label{thm:taylor}
  Let $p(x) \in \KK[x]$. 
  For all $a, b \in \KK[x]$, we have
  $p(a + b) = p(a) + b \cdot \der{p}{x}(a + b)$. 
\end{thm}

  The category of finite powers of $\KK$, $\cat{\KK_\times}$, has all natural numbers $n$ as objects.
  The morphisms $\KK_\times[m, n]$ are $n$-tuples of polynomials $(p_1, \ldots, p_n)$ where
  $p_1, \ldots, p_n \in \KK[x_1, \ldots, x_m]$.
  Composition of morphisms is the usual composition of polynomials.

\begin{rep-lemma}{kleene-change-actions-model}
  The category $\cat{\KK_\times}$ is a cartesian category, endowed with a change action model $\alpha : \cat{\KK_\times} \to \ftor{\CAct}(\cat{\KK_\times})$ whereby $\alpha(\KK) \defeq (\KK, \KK, +, +, 0)$, $\alpha(\KK^i) \defeq \alpha(\KK)^i$; for $\overline p = (p_1(\overline{x}), \ldots, p_n(\overline{x})) : \KK^m \ra \KK^n$,
      $\alpha(\overline p) \defeq (\overline p, (p_1', \ldots, p_n'))$, where 
      \(  p_i' = p_i' (x_1, \ldots, x_m, y_1, \ldots, y_m) \defeq 
        \sum_{j=1}^n y_j \cdot \der{p_i}{x_j}(x_1 + y_1, \ldots, x_m + y_m).
      \)
\end{rep-lemma}
\proofref

\begin{rem}
Interestingly derivatives are not \lochanged{additive} in the second argument.
Take $p(x) = x^2$. 
Then $\d p(a, b+c) > \d p (a, b) + \d p (a, c)$.
It follows that $\KK[\overline x]$ cannot be modelled by GCDC (\lochanged{because of axiom [CD.2]}).
\end{rem}

\map{
  I think there is another, more general construction which understands a ``polynomial'' as a sum of monomials $p(x) = p_0() + p_1(x) + p_2(x, x) + p_3(x, x, x) \ldots$, where one can find a change action model, but this is trickier: 
  the types are more general (a polynomial may go between different spaces) and the behaviours are stranger 
  (if one considers the free monoid on $a, b$, there is a ``generalised polynomial'' that maps $a$ to $b$ and $b$ to $a$, for example).
}

\section{$\omega$-change actions and $\omega$-differential maps}\label{sec:omega-change-actions}

\newcommand{\seq}[1]{[ {#1} ]} 
\newcommand{\changeseq}[1]{\seq{{#1}_i}}
\newcommand{\diffseq}[1]{\seq{{#1}_i}}

A change action model $\alpha : \cat{C} \to \ftor{\CAct}(\cat{C})$ is a category that supports higher-order differentials:
each $\cat{C}$-object $A$ is associated with an $\omega$-sequence of change actions---\lochanged{\(
\alpha(A), \alpha(\changes A), \alpha(\changes^2 A), \ldots
\)}---in which every change action is compatible with the neighbouring change actions.
We introduce \emph{$\omega$-change actions} as a means of constructing change action models ``freely'':
given a cartesian category $\cat{C}$, the objects of the category $\ftor{\CAct_\omega}(\cat{C})$ are all $\omega$-sequences of ``contiguously compatible'' change actions.

\medskip

We work with $\omega$-sequences $\seq{A_i}_{i \in \omega}$ and $\seq{f_i}_{i \in \omega}$ of objects and morphisms in $\cat{C}$.
We write $\pr_k(\seq{A_i}_{i \in \omega}) \defeq  A_k$ for the $k$-th element of the $\omega$-sequence (similarly for $\pr_k(\seq{f_i}_{i \in \omega})$),
and omit the subscript `$i \in \omega$' from $\seq{A_i}_{i \in \omega}$ to reduce clutter.
Given $\omega$-sequences $\seq{A_i}$ and $\seq{B_i}$ of objects of a cartesian category $\cat{C}$, define $\omega$-sequences, \emph{product} $\seq{A_i} \times \seq{B_i}$, \emph{left shift} $\Pi \seq{A_i}$ and \emph{derivative space} $\A{D}\seq{A_i}$, by:
  \begin{align*}
     \pr_j (\seq{A_i} \times \seq{B_i}) 
    & \defeq  
    A_j \times B_j
    \qquad
    \pr_j (\Pi \seq{A_i} ) \defeq  
    A_{j+1}\\
    \pr_0 (\A{D} \seq{A_i}) & \defeq  
    A_0
    \qquad
    \pr_{j+1}\A{D} \seq{A_i} \defeq
    \pr_j \A{D}\seq{A_i} \times \pr_j \A{D}(\Pi \seq{A_i})
  \end{align*}


\begin{exa}
  Given an $\omega$-sequence $\seq{A_i}$, the first few terms of $\A{D}\seq{A_i}$ are:
  \begin{align*}
    \pr_0 \A{D}\seq{A_i} &= A_0 \quad
    \pr_1 \A{D}\seq{A_i} = A_0 \times A_1\quad
    \pr_2 \A{D}\seq{A_i} = (A_0 \times A_1) \times (A_1 \times A_2)\\
    \pr_3 \A{D}\seq{A_i} &= \big((A_0 \times A_1) \times (A_1 \times A_2)\big)
                           \times \big((A_1 \times A_2) \times (A_2 \times A_3)\big)
  \end{align*}
\end{exa}

\begin{defi}\rm
  Given $\omega$-sequences $\seq{A_i}$ and $\seq{B_i}$,
  a \emph{pre-$\omega$-differential map} between them, written $\seq{f_i} : \seq{A_i} \ra \seq{B_i}$,
  is an $\omega$-sequence $\seq{f_i}$ such that for each $j$, $f_j : \pr_j \A{D} \seq{A_i} \ra B_j$ is a $\cat{C}$-morphism.
\end{defi}

We explain the intuition behind the derivative space $\A{D} \seq{A_i}$.
Take a morphism $f : A \to B$, and set $A_i = \changes^i {A}$ (where $\changes^0 \defeq A$ and $\changes^{n+1} A \defeq \changes (\changes^n A)$).
Since $\changes$ distributes over product,
the domain of the $n$-th derivative of $f$ is $\pr_n \A{D} \seq{A_i}$.

\noindent \emph{Notation}. Define $\epow{0} \defeq  \pi_1$ and $\epow{j+1} \defeq  \epow{j} \times \epow{j}$; and define $\cpow{0} \defeq  \Id$ and $\cpow{j+1} \defeq  \pi_2 \circ \cpow{j}$.

\begin{defi}\rm\label{def:3-actions-omega-seq}
  Let $\seq{f_i} : \seq{A_i} \ra \seq{B_i}$ and $\seq{g_i} : \seq{B_i} \ra \seq{C_i}$ be pre-$\omega$-differential maps. 
  The \emph{derivative sequence} $\A{D}\seq{f_i}$ is the $\omega$-sequence defined by:
  \[
    \pr_j \A{D}\seq{f_i} \defeq  \pair{f_j \circ \epow{j}}{f_{j+1}} : \pr_{j+1}\A{D}\seq{A_i} \ra B_j \times B_{j+1}
  \]
  Using the shorthand $\A{D}^n \seq{f_i} \defeq \underbrace{\A{D}(\ldots (\A{D}}_{n\text{ times}}\seq{f_i}))$, the \emph{composite} $\seq{g_i} \circ \seq{f_i} : \seq{A_i} \ra \seq{C_i}$ is the pre-$\omega$-differential map given by
  \(
    \pr_j (\seq{g_i} \circ \seq{f_i}) = g_j \circ \pr_0 (\A{D}^j\seq{f_i}).
  \)
  The \emph{identity} pre-$\omega$-differential map $\Id : \changeseq A \to \changeseq A$ is defined as: 
  \(
    \pr_j \Id \defeq  
    \cpow{j} : \pr_j \A{D} \changeseq A \to A_j.
  \)
\end{defi}

\begin{exa}
Consider $\omega$-sequences $\seq{f_i}$ and $\seq{g_i}$ as above. Then:
  \begin{align*}
    \pr_0 \A{D}\seq{f_i} &= \pair{f_0 \circ \epow{0}}{f_1}
\qquad 
    \pr_1 \A{D}\seq{f_i} = \pair{f_1 \circ \epow{1}}{f_2} 
    \\
    \pr_0 \A{D}^2 \seq{f_i} &= \pair{\pair{f_0 \circ \epow{0}}{f_1} \circ \pi_1}{\pair{f_1 \circ \epow{1}}{f_2}}\\
    \pr_1 \A{D}^2\seq{f_i} &= 
    \pair{
      \pair{f_1 \circ \epow{1}
       }{
       f_2} \circ \epow{1}
    }{
      \pair{f_2 \circ \epow{2} }{f_3}
    }\\
  \pr_0 \A{D}^3 \seq{f_i}
  &= 
\pair{
\pr_0 \A{D}^2\seq{f_i} \circ \epow{0}
  }{
\pair{
  \pair{f_1 \circ \epow{1}}{f_2} \circ \epow{1}
}{
  \pair{f_2 \circ \epow{2}}{f_3}
}  
  }
  \end{align*}
  It follows that the first few terms of the composite $\seq{g_i} \circ \seq{f_i}$ are:
  \begin{align*}
    \pr_0 (\seq{g_i} \circ \seq{f_i}) &= g_0 \circ f_0 \qquad
    \pr_1 (\seq{g_i} \circ \seq{f_i}) = g_1 \circ \pair{f_0 \circ \epow{0}}{f_1}\\
    \pr_2 (\seq{g_i} \circ \seq{f_i}) &=
    g_2 \circ \pair{\pair{f_0 \circ \pi_1}{f_1} \circ \epow{0}}{\pair{f_1 \circ \epow{1}}{f_2}}
  \end{align*}
  Notice that these correspond to iterations of the chain rule, assuming $f_{i+1} = \d f_i$ and $g_{i+1} = \d g_i$.
\end{exa}

\begin{rep-proposition}{id-law}
  For any pre-$\omega$-differential map $\seq{f_i}$, $\Id \circ \seq{f_i} = \seq{f_i} \circ \Id = \seq{f_i}$.
\end{rep-proposition}
\proofref
\begin{rep-proposition}{comp-assoc}
  Composition of pre-$\omega$-differential maps is associative: given pre-$\omega$-differential maps $\seq{f_i} : \changeseq A \to \changeseq B$,
  $\seq{g_i} : \changeseq B \to \changeseq C$ and $\seq{h_i} : \changeseq C \to \changeseq D$,
  then for all $n \geq 0$,
  \(
    h_n \circ \pr_0 \A{D}^n (\seq{g_i} \circ \seq{f_i}) = (h_n \circ \pr_0 \A{D}^n \seq{g_i}) \circ \pr_0 \A{D}^n \seq{f_i}.
  \)
\end{rep-proposition}
\proofref

\begin{defi}\rm
  Given pre-$\omega$-differential maps $\seq{f_i} : \seq{A_i} \ra \seq{B_i}, \seq{g_i} : \seq{A_i} \ra \seq{C_i}$,
  the \emph{pairing} $\pair{\seq{f_i}}{\seq{g_i}} : \seq{A_i} \ra \seq{B_i} \times \seq{C_i} $ is the pre-$\omega$-differential map defined by:
  \(
    \pr_j \pair{\seq{f_i}}{\seq{g_i}} = \pair{f_j}{g_j}.
  \)
  Define pre-$\omega$-differential maps $\pi_{\mathbf 1} \defeq \seq{{\pi_{\mathbf 1}}_i} : \seq{A_i} \times \seq{B_i} \ra \seq{A_i}$ by $\pr_j \seq{{\pi_{\mathbf 1}}_i} \defeq \pi_1 \circ \cpow{j}$, and $\pi_{\mathbf 2} \defeq \seq{{\pi_{\mathbf 2}}_i} : \seq{A_i} \times \seq{B_i} \ra \seq{B_i}$ by $\pr_j \seq{{\pi_{\mathbf 2}}_i} \defeq \pi_2 \circ \cpow{j}$.
\end{defi}

\begin{defi}\rm
  \label{def:pre-omega-change-action}
  A \emph{pre-$\omega$-change action} on a cartesian category $\cat{C}$ is a quadruple 
  \(
    \ds{A} = (\seq{A_i}, \seq{\ds{\oplus^A}_i}, \seq{\ds{+^A}_i}, \seq{0^A_i})
  \)
  where $\changeseq A$ is an $\omega$-sequence of $\cat C$-objects, and for each $j \geq 0$, $\ds{\oplus^A}_j$ and $\ds{+^A}_j$ are $\omega$-sequences, satisfying
  \begin{enumerate}
    \item 
    $\ds{\oplus^A}_j : \Pi^j \changeseq{A} \times \Pi^{j+1}\changeseq{A} \to \Pi^j\changeseq{A}$ is a pre-$\omega$-differential map.
    \item 
    $\ds{+^A}_j : \Pi^{j + 1} \changeseq{A} \times \Pi^{j+1}\changeseq{A} \to \Pi^{j+1}\changeseq{A}$ is a pre-$\omega$-differential map.
    \item $0^A_j : \top \to A_{j+1}$ is a $\cat{C}$-morphism.
    \item $\changes(\ds{A}, j)\defeq  (A_j, A_{j+1}, \pr_0 \ds{\oplus^A}_j, \pr_0 \ds{+^A}_j, {0^A_j})$ is a change action in $\cat{C}$.
  \end{enumerate}
  We extend the left-shift operation to pre-$\omega$-change actions by defining
  $\Pi \ds{A} \defeq ( \Pi\seq{A_i}, \Pi \seq{\ds{\oplus^A}_i}, \Pi \seq{\ds{+^A}_i}, \seq{0_i^A})$.
  Then we define the \lochanged{change actions} $\A{D}(\ds{A}, j)$ inductively by:
  $\A{D}(\ds{A}, 0) \defeq \changes(\ds{A}, 0)$ and
  $\A{D}(\ds{A}, j+1) \defeq \changes(\ds{A}, j) \times \changes(\Pi\ds{A}, j)$. 
  Notice that the carrier object of $\A{D}(\ds{A}, j)$ is the $j$-th element of the $\omega$-sequence $\A{D}\seq{A_i}$.
\end{defi}

\begin{defi}\rm
  Given pre-$\omega$-change actions $\ds A$ and $\ds B$ (using the preceding notation), a pre-$\omega$-differential map $\seq{f_i} : \changeseq{A} \to \changeseq{B}$ is \emph{$\omega$-differential} if, for each $j \geq 0$, $(f_j, f_{j+1})$ is a differential map from the change action $\A{D}(\ds{A}, j)$ to $\changes(\ds{B}, j)$. 
  Whenever $\seq{f_i}$ is an $\omega$-differential map, we write $\ds{f} : \ds{A} \ra \ds{B}$.
  
  We say that a pre-$\omega$-change action $\ds{A}$ is an \emph{$\omega$-change action} if, for each $i \geq 0$, $\ds{\oplus^A}_i$ and $\ds{+^A}_i$ are $\omega$-differential maps.
\end{defi}

\begin{rem}
It is important to sequence the definitions appropriately.
Notice that we only define $\omega$-differential maps once there is a notion of pre-$\omega$-change action, 
but pre-$\omega$-change actions need pre-$\omega$-differential maps to make sense of the monoidal sum $\ds{\plus}_j$ and action $\ds{\oplus}_j$.
\item \lochanged{The reason for requiring each $\ds{\oplus^A}_i$ and $\ds{+^A}_i$ in an $\omega$-change object 
\(
    \ds{A} = (\seq{A_i}, \seq{\ds{\oplus^A}_i}, \seq{\ds{+^A}_i}, \seq{0^A_i})
\)
to be $\omega$-differential is so that $\ds A$ is \emph{internally} a change action in $\ftor{\CAct_\omega}(\cat{C})$ (see Def.~\ref{def:cact-omega-cat}).}
\end{rem}

\begin{rep-lemma}{composite-omega-differential}
  Let $\ds{f} : \ds{A} \ra \ds{B}$ and $\ds{g} : \ds{B} \ra \ds{C}$ be $\omega$-differential maps.
  {Qua} pre-$\omega$-differential maps, their composite $\seq{g_i} \circ \seq{f_i}$ is $\omega$-differential.
  Setting $\ds{g} \circ \ds{f} \defeq \seq{g_i} \circ \seq{f_i} : \ds A \to \ds C$, it follows that composition of $\omega$-differential maps is associative.
\end{rep-lemma}

\proofref

\begin{rep-lemma}{id-omega-differential}
  For any $\omega$-change action $\ds{A}$, the pre-$\omega$-differential map 
  ${\Id} : \changeseq{A} \ra \changeseq{A}$
  is $\omega$-differential. 
  Hence $\ds{\Id} := \Id : \ds{A} \to \ds{A}$ satisfies the identity laws.
\end{rep-lemma}
\proofref

\lochanged{
\begin{defi}\rm
  \label{def:cact-omega-products}
  Given $\omega$-change actions $\ds{A}$ and $\ds{B}$, we define the \emph{product $\omega$-change action} by:
  \(
    \ds{A} \times \ds{B} \defeq ( 
    \seq{A_i \times B_i}, 
    \seq{\ds{\oplus'}_i}, 
    \seq{\ds{+'}_i}, 
    \seq{0'_i}
    )
  \)
where
\map{Intuitively: same as in $\CAct$, except $\omega$-differential maps}
\begin{enumerate}
  \item $\ds{\oplus'}_j \defeq \pair{\ds{\oplus^A}_j}{\ds{\oplus^B}_j} 
    \circ \pair{\pair{\ds{\pi_{11}}}{\ds{\pi_{12}}}}
               {\pair{\ds{\pi_{21}}}{\ds{\pi_{22}}}}$
  \item $\ds{+'}_j \defeq \pair{\ds{+^A}_j}{\ds{+^B}_j} 
    \circ \pair{\pair{\ds{\pi_{11}}}{\ds{\pi_{12}}}}
               {\pair{\ds{\pi_{21}}}{\ds{\pi_{22}}}}$
  \item $0_j' \defeq \pair{0^A_j}{0^B_j}$
\end{enumerate}
Notice that $\changes(\ds{A} \times \ds{B}, j)\defeq  (A_j \times B_j, A_{j+1} \times B_{j+1}, \pr_0 \ds{\oplus'}_j, \pr_0 \ds{+'}_j, {0_j'})$ is a change action in $\cat{C}$ by construction.
\end{defi}
}

\begin{rep-lemma}{cact-omega-products}
  The pre-$\omega$-differential maps $\pi_{\mathbf 1}, \pi_{\mathbf 2}$ are $\omega$-differential.
  Moreover, for any $\omega$-differential maps $\ds f : \ds{A} \ra \ds{B}$ and $\ds{g} : \ds{A} \ra \ds{C}$,
  the map $\pair{\ds f}{\ds g} \defeq \pair{\seq{f_i}}{\seq{g_i}}$ is $\omega$-differential, \lochanged{satisfying}
    \(
    \ds{\pi_{\mathbf 1}} \circ \pair{\ds f}{\ds g} = \ds f
    \) and
    \(
    \ds{\pi_{\mathbf 2}} \circ \pair{\ds f}{\ds g} = \ds g.
    \)
\end{rep-lemma}
\proofref

\begin{defi}\rm\label{def:cact-omega-cat}
  Define the functor $\ftor{\CAct_\omega} : \cart \ra \cart$ as follows.
\begin{itemize}
  \item $\ftor{\CAct_\omega}(\cat{C})$ is the category whose objects are the $\omega$-change
    actions over $\cat{C}$ and whose morphisms are the $\omega$-differential maps.
  \item If $\ftor{F} : \cat{C} \ra \cat{D}$ is a (product-preserving) functor, then
    $\ftor{\CAct_\omega}(\ftor{F}) : \ftor{\CAct_\omega}(\cat{C}) \ra \ftor{\CAct_\omega}(\cat{C})$
    is the functor mapping the $\omega$-change action 
        \(
          \left(\changeseq{A}, \seq{\diffseq{\oplus}_j}, \seq{\diffseq{+}_j}, \seq{0_j}\right)
        \)
        to 
        \(
          \left(\changeseq{\ftor{F}A}, \seq{\diffseq{\ftor{F}\oplus}_j},
          \seq{\diffseq{\ftor{F}+}_j}, \seq{\ftor{F}0_j}\right)
        \);
      and the $\omega$-differential map $\diffseq{f}$ to 
        $\diffseq{\ftor{F}f}$.
\end{itemize}
\end{defi}

\begin{rep-theorem}{cact-omega-ccc}
  The category $\ftor{\CAct_\omega}(\cat{C})$ is cartesian, with product given in Def.~\ref{def:cact-omega-products}.
  Moreover if $\cat{C}$ \lochanged{is closed and} has countable limits, $\ftor{\CAct_\omega}(\cat{C})$ is cartesian closed.
\end{rep-theorem}
\proofref

\newcommand\cseq[1]{\Delta^\omega_{#1}}

\begin{rep-theorem}{cact-omega-canonical}
  The category $\ftor{\CAct_{\omega}}(\cat{C})$ is equipped with a canonical change action model:
  $\gamma : \ftor{\CAct_\omega}(\cat{C}) \ra \ftor{\CAct}(\ftor{\CAct_\omega}(\cat{C}))$.
\end{rep-theorem}
\begin{thm}[Relativised final coalgebra]
Let $\cat{C}$ be a change action model.
The canonical change action model $\gamma : \ftor{\CAct_\omega}(\cat{C}) \ra \ftor{\CAct}(\ftor{\CAct_\omega}(\cat{C}))$ is a 
relativised\footnote{Here $\ftor{\CAct}$ is restricted to the full subcategory of $\cat{Cat}_\times$ with $\cat{C}$ as the only object.} 
final coalgebra of $(\ftor{\CAct}, \epsilon)$. 

I.e.~for all change action models on $\cat{C}$, $\alpha : \cat{C} \to \ftor{\CAct}(\cat{C})$, there is a unique coalgebra homomorphism $\alpha_\omega : \cat{C} \to \ftor{\CAct_\omega}(\cat{C})$, as witnessed by the commuting diagram:
\begin{center}
  \begin{tikzcd}
      \cat{C}
      \arrow[d, "\exists \, ! \, \alpha_\omega", swap]
      \arrow[r, "\alpha"]
      &\ftor{\CAct}(\cat{C})
      \arrow[d, "{\ftor{\CAct}(\alpha_\omega)}"]
      \\
      \ftor{\CAct_\omega}(\cat{C})
      \arrow[r, "\gamma"]
      & \ftor{\CAct}(\ftor{\CAct_\omega}(\cat{C}))
  \end{tikzcd}
\end{center}
\end{thm}

\lochanged{
\begin{proof}
We first exhibit the functor $\alpha_\omega : \cat{C} \to \ftor{\CAct_\omega}(\cat{C})$. 

Take a $\cat{C}$-morphism $f : A \to B$.
We define the $\omega$-differential map $\alpha_\omega(f) \defeq \ds f : \ds A \to \ds B$, where $\ds{A} \defeq \left(\seq{A_i}, \seq{\ds{\oplus}_i}, \seq{\ds{+}_i}, \seq{0_i}\right)$ is the $\omega$-change action determined by $A$ under \emph{iterative actions of $\alpha$}.
I.e.~for each $i \geq 0$: $A_i \defeq  \changes^i{A}$ (by abuse of notation, we write $\changes A'$ to mean the carrier object of the monoid of the internal change action $\alpha(A')$, for any $\cat{C}$-object $A'$); 
$\ds{\oplus}_j : \Pi^j \changeseq{A} \times \Pi^{j+1}\changeseq{A} \to \Pi^j\changeseq{A}$ is specified by: $\pr_k \ds{\oplus}_j$ is the monoid action morphism of $\alpha(A_{j+k})$; 
$\ds{+}_j : \Pi^{j+1}\changeseq{A} \times \Pi^{j+1}\changeseq{A} \to \Pi^{j+1}\changeseq{A}$ is specified by: $\pr_k \ds{\oplus}_j$ is the monoid sum morphism of $\alpha(A_{j+k})$;
$0_i$ is the zero object of $\alpha(A_i)$.

The $\omega$-sequence $\ds f \defeq \seq{f_i}$ is defined by induction: 
$f_0 \defeq  f$; assume $f_n : (\A{D} \ds A)_n \to B_n$ is defined and suppose $\alpha(f_n) = (f_n, \d f_n)$ then define $f_{n+1} \defeq  \d f_n$.

To see that the diagram commutes, notice that $\gamma(\ds f) = (\ds f, \Pi \ds f)$ and $\ftor{\CAct}(\alpha_\omega)$ maps $\alpha(f) = (f, \d f)$ to $(\ds f, \ds {\d f})$; then observe that $\Pi \ds f = \ds {\d f}$ follows from the construction of $\ds f$.

Finally to see that the functor $\alpha_\omega$ is unique, consider the $\cat{C}$-morphisms $\d^n f$ $(n = 0, 1, 2, \cdots)$ where $\alpha(\d^n f) = (\d^n f, \d^{n+1} f)$. 
Suppose $\beta : \cat{C} \to \ftor{\CAct_\omega}(\cat{C})$ is another homomorphism.
Thanks to the commuting diagram, we must have $\Pi^n \beta(f) = \beta(\d^n f)$, and so, in particular $(\beta(f))_n = (\Pi^n \beta(f))_0 = (\beta(\d^n f))_0 = \d^n f$, for each $n \geq 0$. 
Thus~$\ds f = \beta(f)$ as desired.
\qed
\end{proof}
}

Intuitively any change action model on $\cat{C}$ is always a ``subset'' of the change action model on
$\ftor{\CAct_\omega}(\cat{C})$.

\begin{rep-theorem}{cact-as-limit}
  The category $\ftor{\CAct_\omega}(\cat{C})$ is the limit in $\cart$ of the diagram
  \begin{center}
    \begin{tikzcd}[ampersand replacement=\&]
      \cat{D}
      \arrow[dd]
      \arrow[ddr]
      \arrow[ddrr]
      \\\\
      \ftor{\CAct}(\cat{C})
      \& \ftor{\CAct}(\ftor{\CAct}(\cat{C}))
      \arrow[l, shift left, "\xi"] 
      \arrow[l, shift right, swap, "\epsilon"]
      \& \ftor{\CAct}(\ftor{\CAct}(\ftor{\CAct}(\cat{C})))
      \arrow[l, shift left, "\xi"] 
      \arrow[l, shift right, swap, "\epsilon"]
      \& \ldots \arrow[l, shift left, "\xi"]
      \arrow[l, shift right, swap, "\epsilon"]
    \end{tikzcd}
  \end{center}
\end{rep-theorem}
\proofref

\lo{The intuition is that $\ftor{\CAct_\omega}(\cat{C})$ is a model of higher-order derivatives. 

- What ``axioms'' of (higher-order) derivatives are satisified by $\ftor{\CAct_\omega}(\cat{C})$?

- If we assume the strictness conditions of generalised cartesian differential categories, does $\ftor{\CAct_\omega}(\cat{C})$ satisfy any of the axioms [CD.1-7]?

- Is it reasonable to expect they do?}

\map{
  Not sure what strictness means in this context.}
\lo{By \emph{strictness conditions} in the definition of generalised differential cartesian categories, I mean the following: commutativity of the monoid $L(A)$, $L(A \times B) = L(A) \times L(B)$, and $L(L_0(A)) = L(A)$.}
\map{
  $\ftor{\CAct_\omega}(\cat{C})$ satisfies CD.3-5, for the rest additional axioms are necessary.
  There should be `smaller' versions of $\ftor{\CAct_\omega}$ (only tight change actions, only commutative change actions (the original definition actually
  was restricted to tight, commutative change actions), only stable derivatives...) that allow for more or less the same results but satisfy more
  axioms.
}

\lo{Comment on (any) connections between the $\ftor{\CAct_\omega}(\cat{C})$ construction and the Fa\`a di Bruno comonad.}

\section{Related work, future directions and conclusions}

Firstly, the present work directly expands upon work by the authors and others in \cite{mario2019fixing}, where
the notion of change action was developed in the context of the incremental evaluation of Datalog programs.
This work generalizes some results in \cite{mario2019fixing} and addresses two significant questions that had
been left open, namely: how to construct cartesian closed categories of change actions and how to formalize
higher-order derivatives.

Our work is also closely related to Cockett, Seely and Cruttwell's work on cartesian differential categories
\cite{blute2009cartesian,blute2010convenient,cockett2014differential} 
and Cruttwell's more recent work on generalised cartesian differential categories \cite{cruttwell2017cartesian}.
Both cartesian differential categories and change action models aim to provide a setting for differentiation, and the construction of $\omega$-change actions resembles the Fa\`a di Bruno construction \cite{cruttwell2017cartesian,CockettS11} 
(especially its recent reformulation by Lemay\footnote{In a notable case of parallel evolution, we developed our results independently of Lemay.} \cite{Lemay18})
which, given an arbitrary category $\cat{C}$, builds a cofree cartesian differential category for it).
The main differences between these two settings lie in the specific axioms required (change action models are significantly weaker: see Remark~\ref{rem:GCDC-axioms}) and the approach taken to define derivatives. 

This last point is of particular interest: cartesian differential categories 
assume there is a notion of differentiation in place that satisfies certain coherence conditions, 
whereas change actions give an algebraic definition of what it means for a function to have a derivative in terms of the change action structure.
In this sense, the derivative condition is close to the Kock-Lawvere axiom from synthetic differential geometry \cite{Kock06,lavendhomme2013basic}, which has provided much of the driving intuition behind this work, and making this connection precise is the subject of ongoing research.

In a different direction, the nice interplay between derivatives, products and exponentials in
closed change action models (see Theorem~\ref{thm:oplus-lambda-infinitesimal}) suggests that there should be a reasonable calculus for change action models. 
It would be interesting to formulate such a calculus and explore its connections to the differential $\lambda$-calculus \cite{DBLP:journals/tcs/EhrhardR03}. 
This could also lead to practical applications to languages for incremental computation, or even higher-order automatic differentiation \cite{kelly2016evolving}.

In conclusion, change actions and change action models constitute a new setting for reasoning about differentiation that is able to unify ``discrete'' and ``continuous'' models, as well as higher-order functions.
Change actions are remarkably well-behaved and show tantalising connections with geometry and 2-categories.
We believe that most ad hoc notions of derivatives found in disparate subjects can be elegantly integrated into the framework of change action models.
We therefore expect any further work in this area to have the potential of benfitting these notions of derivatives.

\lo{It is important to justify the notion of change action and change action model, and give convincing examples of change action models (e.g.~models of incremental computation, and discrete calculus) which are \emph{not} (generalised) cartesian differential categories.}

\map{
  This is a shame. I wonder if Cruttwell might be more open to the idea, especially since he was the one to generalise cartesian differential categories.
  (I wonder whether this is the reason why he wrote that paper without Cockett)
}

\lo{You are right that the connections with Cruttwell's work (generalised cartesian different categories) are more the issue here. 
I don't know Cruttwell. I hope he is a reasonable person.}

\todo{Discuss:
- Lemay's latest work \cite{Lemay18}.

- Tangent categories}

\lo{22 Oct 2018.
What are the higher-order derivatives of the identity morphism?
This is related to the definition of the identity morphism and morphism composition of the category $\ftor{\CAct_\omega}(\cat{C})$.

Assuming thinness of the change structure $A$, we must have $\d \Id (a, \change{a}) = \change a$.

Similarly, assuming thiness of the change structure on $A \times \changes A$, and $\changes{(A \times \changes{A})} = \changes{A} \times \changes{\changes{A}}$, we must have 
$\d^2 \Id ((a, \change{a}), (\change{x}, \change{\change{x}})) = \change{\change{x}}$. 
Note that $\d^2 \Id$ is not $0$.
\map{This is actually as in cartesian differential categories.
A CDC in principle does not say anything about $\DF{\Id}$, but when the object is of the form $A \times B$ then
$\Id = \pair{\pi_1}{\pi_2}$ and then, according to the rules of a CDC, we must have $\DF{\Id_{A \times B}} = \pi_2$, just
as is the case for change actions. Since $\DF{\pi_1} = \pi_1 \circ \pi_2$, the second derivative also coincides.}

This contrasts with the Fa\`a di Bruno construction in \cite{cruttwell2017cartesian}.
In Cruttwell's treatment, given a morphism $f : A \to B$, its $n$-th derivative is a morphism 
\[
f_n : A \times \underbrace{\changes{A} \times \cdots \times \changes{A}}_n \to \changes{B}
\] 
as opposed to $A \times (\changes{A})^{n+1} \to \changes{B}$.
Recall that they assume, crucially: (i) $\changes$ distributes over product, i.e., $\changes{(A \times B)} = \changes{A} \times \changes{B}$, 
and (ii) idempotence of $\changes$, i.e., $\changes{\changes{A}} = \changes{A}$.
Note that in their setting second- and higher-order derivatives of the identity morphism is the 0-morphism.

\map{This isn't entirely true. 
It's true that in the Fa\`a di Bruno construction the $n$-th term of the sequence $(f_0, f_1, \ldots)$ that defines the identity map has $f_{n+2} = 0$, 
but this is because the Fa\`a di Bruno construction only keeps track of the non-linear part of the derivative.

JS Lemay has a much simpler treatment of the Fa\`a di Bruno construction (``A Tangent Category Alternative to the Fa\`a di Bruno Construction'' in arXiv) 
given in terms of what he calls D-sequences, which is almost identical to my own $\omega$-differential maps (he even has a notion of pre-D-sequence for a sequence of morphisms ``of the right type'' that don't necessarily satisfy the axioms, just like I do. 
Interestingly, we seem to have developed these notions in parallel, but completely independently).
}

Morphisms of the category are infinite sequences $(f_\ast, f_1, f_2, \cdots)$ where $f_\ast : A \to B$, and for each $i$, $f_i$ is assumed to be additive and symmetric (it is intuitively the $i$-th order derivative of $f_\ast$).
The key idea behind the Fa\`a di Bruno construction is a definition of composition of morphisms.

Questions: are there examples / situations of differentiation whereby (i) and (ii) above do not hold?

}

\bibliographystyle{plainnat}
\bibliography{paper}

\clearpage
\appendix


\section{Supplementary materials for Section~\ref{sec:change-actions}}

\begin{proof-for-prop}{unique-derivative-regular}
  Suppose $a \in \us{A}$ (if $\us{A}$ is empty the property follows trivially), and $\change{a}, \change{b} \in \changes{A}$.
  Then
  \begin{align*}
    f(a \oplus (\change{a} + \change{b}))
    &= f(a \oplus \change{a} \oplus \change{b})
    \\
    &= f(a \oplus \change{a}) \oplus \d f(a \oplus \change{a}, \change{b})
    \\
    &= \big(f(a) \oplus \d f(a, \change{a})\big) \oplus \d f(a \oplus \change{a}, \change{b})
    \\
    &= f(a) \oplus \big(\d f(a, \change{a}) + \d f(a \oplus \change{a}, \change{b})\big)
  \end{align*}
  Thus we can define the following derivative for $f$
  $$
  \d f_a(x, \change{x}) \defeq
  \begin{cases}
    \d f(a, \change{a}) + \d f(a \oplus \change{a}, \change{b})
    &\text{ when $x = a, \change{x} = \change{a} + \change{b}$}\\
    \d f(x, \change{x})
    &\text{ otherwise}
  \end{cases}
  $$
  Since the derivative is unique, it must be the case that $\d f = \d f_a$ and therefore
  $\d f(a, \change{a} + \change{b}) = \d f(a, \change{a}) + \d f(a \oplus \change{a}, \change{b})$.
  By a similar argument, $\d f(a, 0) = 0$ and thus $\d f$ is regular.
  
\end{proof-for-prop}

\begin{rem}\label{rem:reg-der-not-nec}
  One may wonder whether every differentiable function admits a regular derivative: the answer is no. 
  Consider the change actions:
  \begin{align*}
    \da{A_1} &= (\ZZ_2, \ZZ_2, \plus, \plus, 0)\qquad
    \da{A_2} = (\ZZ_2, \NN, \left[ \plus \right], \plus, 0)
  \end{align*}
  where $\left[ m \right]\left[ \plus \right] n = \left[ m \plus n \right]$. 
  The
  identity function $\Id : \da{A_1} \ra \da{A_2}$ admits infinitely many
  derivatives, none of which are regular. The condition under which a 
  (differentiable) function admits a regular derivative is an open question.
\end{rem}

\begin{proof-for-prop}{chain-rule-regular}
\[
  \begin{array}{ll}
    & (\d g \circ \pair{f \circ \pi_1}{\d f})(a, \change{a} \plus \change{b})\\ 
    = \ & \d g(f(a), \d f(a, \change{a} \plus \change{b}))\\
    =& \d g(f(a), \d f(a, \change{a}) \plus \d f(a \oplus\change{a}, \change{b}))\\
    =& \d g(f(a), \d f(a, \change{a}))
     \plus \d g(f(a) \oplus \d f(a, \change{a}), \d f(a\oplus\change{a}, \change{b}))\\
    =& (\d g \circ \pair{f \circ \pi_1}{\d f})(a, \change{a})
     \plus (\d f \circ \pair{f\circ \pi_1}{\d f})(a \oplus \change{a}, \change{b})
  \end{array}
\]
  \begin{align*}
    (\d g \circ \pair{f \circ \pi_1}{\d f})(a, 0)
    &= \d g(f(a), \d f(a, 0))\\
    &= \d g(f(a), 0)\\
    &= 0
  \end{align*}
  
\end{proof-for-prop}

\begin{proof-for-theorem}{cact-preord}

  Consider an arbitrary change action $A = (\us A, \changes A, \oplus, +, 0)$. Its structure as a change action
  induces a natural preorder on the base set $\us A$.
  \begin{defi}[Reachability preorder]\rm
    For $a, b \in \us A$, we define $a \reachOrder b$ iff there is a $\change{a} \in \changes{A}$ such that
    $a \oplus \change{a} = b$. Then $\reachOrder$ defines a preorder on $\us A$.
  \end{defi}

  The intuitive significance of the reachability preorder induced by $A$ is that it contains all the information about differentiability of functions from or into $A$.
  This is made precise in the following result:

  \begin{lem}
    A function $f : \us A \ra \us B$ is differentiable as a function from $A$ to $B$ iff it is monotone with respect to the reachability preorders on $A, B$.
  \end{lem}
  \begin{proof}
    Let $f$ be a differentiable function, with $\d f$ an arbitrary derivative, and suppose $a \reachOrder_A a'$. Hence there is some $\change a$ such that $a \oplus \change a = a'$.
    Then, by the derivative property, we have $f(a') = f(a \oplus \change a) = f(a) \oplus \d f (a, \change a)$, hence $f(a) \reachOrder_B f(a')$.

    Conversely, suppose $f$ is monotone, and pick arbitrary $a \in A, \change a \in \changes A$. Since $a \reachOrder_A a \oplus \change a$ and $f$ is monotone,
    we have $f(a) \reachOrder_B f(a \oplus \change a)$ and therefore there exists a $\change b \in \changes B$ such that $f(a) \oplus \change b = f(a \oplus \change a)$.
    We define $\d f(a, \changes a)$ to be precisely such a change $\change b$
    (note that the process of arbitrarily picking a $\change b$ for every pair $a, \change a$ makes use, in general, of the Axiom of Choice).
    
  \end{proof}

  The correspondence between a change action and its reachability preorder gives rise to a (full and faithful) functor $\ftor{Reach} : \cat{CAct} \ra \cat{PreOrd}$ that acts as the
  identity on morphisms.

  Conversely, any preorder $\leq$ on some set $\us A$ induces a change action
  \[ A_\leq \defeq (\us A, \us A_\bot, \sqcup, \sqcup, \bot) \]
  where $\us A_\bot$ is the set $\us A$ extended with a bottom element $\bot$, and $\sqcup$ denotes the least upper bound according to the preorder $\leq$.
  Note that the reachability preorder of the change action $A_\leq$ is precisely $\leq$.

  This defines another full and faithful functor $\ftor{Act} : \cat{PreOrd} \ra \cat{CAct}$ that is the identity on morphisms.

  It remains to check that there are natural isomorphisms $\vf{U} : \ftor{Act} \circ \ftor{Reach} \ra \Id_{\cat{CAct}}$
  and $\vf{V} : \ftor{Reach} \circ \ftor{Act} \ra \Id_{\cat{PreOrd}}$. But these are trivial: it suffices to set $\vf{U}_A = \Id_A$
  and $\vf{V}_{(\us A, \leq)} = \Id_{(\us A, \leq)}$. Hence $\ftor{Act}, \ftor{Reach}$ establish an equivalence of categories between $\cat{PreOrd}$ and $\cat{CAct}$.
  
  \map{Not sure if I should explain why the identities are isomorphisms - it was obvious in the POPL paper due to the background material} 
  \lo{Yes we should explain in the arXiv version, post FoSSaCS19 submission.}
\end{proof-for-theorem}

\subsection*{Adjunctions with $\cat{Set}$}

\begin{proof-for-lemma}{forgetful-functor-is-right-adjoint}
  \begin{align*}
    \epsilon\mathcal{D} \circ \mathcal{D}\eta
    &= (\A{Id}, 0) \circ (\eta, !)\\
    &= (\A{Id} \circ \eta, 0)\\
    &= (\A{Id} \circ \A{Id}, 0)\\
    &= (\A{Id}, 0)\\
    &= \A{Id}
  \end{align*}
  Furthermore:
  \begin{align*}
    \mathcal{F}(\epsilon) \circ \eta\mathcal{F}
    &= \mathcal{F}(\A{Id}, 0) \circ \A{Id}\\
    &= \A{Id} \circ \A{Id}\\
    &= \A{Id}
  \end{align*}
\end{proof-for-lemma}

\section{Supplementary materials for Section~\ref{sec:change-actions-arbitrary}}

\subsection*{The category $\ftor{\CAct}(\cat{C})$}

\begin{thm} \label{cact-preserves-products}
  The functor $\ftor{\CAct}$ preserves all products.
\end{thm}

\begin{proof}
  Consider an $I$-indexed product of categories $\prod_{i \in I} \cat{C}_i$. 
  An object $A$ of $\ftor{\CAct}(\prod_{i \in I} \cat{C}_i)$ is a change action $(\us A, \changes A, \oplus_A, +_A, 0_A)$ where:
  \begin{itemize}
    \item $\us A, \changes A$ are $I$-indexed families of objects $\us A_i, \changes A_i$ of $\cat{C}_i$
    \item $\oplus_A$ is an $I$-indexed family of $\cat{C}_i$-morphisms $\oplus_i : \us A_i \times \changes A_i \ra \us A_i$
    \item $+_A$ is an $I$-indexed family of $\cat{C}_i$-morphisms $+_i : \changes A_i \times \changes A_i \ra \us \changes A_i$
    \item $0_A$ is an $I$-indexed family of $\cat{C}_i$-morphisms $0_i : \top_i \ra \changes A_i$
  \end{itemize}
  satisfying the relevant conditions ($(\changes A, +_A, 0_A)$ is a monoid, $\oplus_A$ is an action).
  But this entails that, for every $i$, the triple $(\changes A_i, +_i, 0_i)$ defines a monoid in $\cat{C}_i$, and $\oplus_i$ is an action of this monid on $\ds A_i$. 
  Hence we obtain an $I$-indexed family $A_i$ of change actions in $\ftor{\CAct}(\cat{C}_i)$ respectively.
  Conversely, given any such family, we can always construct the corresponding change action in $\ftor{\CAct}(\prod_{i \in I} \cat{C}_i)$.

  A similar argument applies to differential maps: every differential map $f : A \ra B$ in $\ftor{\CAct}(\prod_{i \in I} \cat{C}_i)$ corresponds to a family
  of differential maps $f_i : A_i \ra B_i$ in $\ftor{\CAct}(\cat{C}_i)$ and vice versa. Hence the functor $\ftor{\CAct}$ preserves all products.
  
\end{proof}

\subsection*{Products and coproducts in $\ftor{\CAct}(\cat{C})$}

\begin{proof-for-theorem}{cact-products}
  Given any pair of differential maps $f_1 : C \ra A, f_2 : C \ra B$, define 
  \begin{gather*}
    \pair{f_1}{f_2} \defeq (\pair{\us{f_1}}{\us{f_2}},
    \pair
    {\d f_1 \circ \pair{\pi_1 \circ \pi_1}{\pi_1 \circ \pi_2}}
    {\d f_2 \circ \pair{\pi_2 \circ \pi_1}{\pi_2 \circ \pi_2}})
  \end{gather*}
  Then $\overline{\pi_i} \circ \pair{f_1}{f_2} = f_i$.
  Furthermore, given any map $h : C \ra A \times B$ whose projections coincide with
  $\pair{f_1}{f_2}$, by applying the universal property of the product in $\cat{C}$, we obtain $h = \pair{f_1}{f_2}$.
  
\end{proof-for-theorem}

\begin{proof-for-theorem}{cact-terminal}
  It is straightforward to check that, given a change action $A$, there is exactly one differential map $\overline{!} \defeq (!, !) : A \ra \overline{\top}$.
  Now given a differential map $(\us f, \d f) : \overline{\top} \ra A$, applying regularity we obtain:
  \begin{align*}
    \d f &= \d f \circ \pair{\Id_\top}{\Id_\top}\\
         &= \d f \circ \pair{\Id_\top}{0_\top}\\
         &= 0_A
  \end{align*}
  
\end{proof-for-theorem}

\begin{proof-for-lemma}{partial-derivatives}
  Abusing the notation again, the lemma is a direct consequence of regularity:
  \begin{align*}
    \d f((a, b), (\change a, \change b)) 
    &= \d f ((a, b), (\change a, 0_B) +_{A \times B} (0_A, \change b))\\
    &= \d f ((a, b), (\change a, 0_B)) 
      +_C \d f ((a, b) \oplus_{A \times B} (\change a, 0_B), (0_A, \change b))\\
    &= \d f ((a, b), (\change a, 0_B))
      +_C \d f ((a \oplus_A \change a, b), (0_A, \change b))
  \end{align*}
  
\end{proof-for-lemma}

\begin{proof-for-theorem}{cact-coproducts}
  Given any pair of differential maps $f_1 : A \ra C, f_2 : B \ra C$ define 
  \begin{align*}
    \left[ \da {f_1}, \da{f_2} \right] &\defeq (\left[ \us{f_1}, \us{f_2} \right], \d h)\\
    \d h &\defeq \left[ \d f_1 \circ (\Id_A \times \pi_1), \d f_2 \circ (\Id_B \times \pi_2) \right] \circ \delta_{A, B, C}
  \end{align*}
  where $\delta_{A, B, C} : (A \sqcup B) \times C \ra (A \times C) \sqcup (B \times C)$ is the distributive law of $\cat{C}$.

  We check that, indeed, the relevant diagram commutes since:
  \[
  \begin{array}{ll}
    &\left[ f_1, f_2 \right] \circ \overline{\iota_1}\\
    =& (\left[ \us{f_1}, \us{f_2} \right], \d h) \circ (\iota_1, \pair{\pi_2}{0_B})\\
    =& (\left[ \us{f_1}, \us{f_2} \right] \circ \iota_1,
      \d h \circ \pair{\iota_1 \circ \pi_1}{\pair{\pi_2}{0_B}})\\
    =& (\us{f_1},
      \left[ \d f_1 \circ (\Id_A \times \pi_1), \d f_2 \circ (\Id_B \times \pi_2) \right]
      \circ \delta_{A, B, C} \circ \pair{\iota_1 \circ \pi_1}{\pair{\pi_2}{0_B}})\\
    =& (\us{f_1},
      \d f_1 \circ (\Id_A \times \pi_1) \circ \pair{\pi_1}{\pair{\pi_2}{0_B}})\\
    =& (\us{f_1},
      \d f_1 \circ \pair{\pi_1}{\pi_2}) \\
    =& (\us{f_1}, \d f_1)\\
    =& f_1
  \end{array}
  \]
  The universal property of the coproduct in $\cat{C}$ entails that if $h = (\us{h}, \d h)$ is such that $h \circ \overline{\iota_i} = f_i$, then $h = \left[ {f_1}, {f_2} \right]$. 
  Furthermore, since 
  \[
  (\us A \sqcup \us B) \times \changes A \times \changes B \ \cong \ (\us A \times \changes A \times \changes B) \sqcup (\us B \times \changes A \times \changes B),
  \]
  the universal property of the coproduct also shows that necessarily $\d h = \d f$.
  
\end{proof-for-theorem}

\subsection*{Change actions as Lawvere theories}

\begin{proof-for-lemma}{delta-algebras-are-t-algebras}
  Consider a natural transformation $\phi : F \circ \mathbb{T} \ra G \circ
  \mathbb{T}$, and define
  \begin{gather*}
    f = \phi_{X} : F(X) \ra G(X)\\
    f' = \pi_2 \circ \phi_{\A{T}X} : F(X) \times F(\changes{X}) \ra G(\changes{X})
  \end{gather*}
  Then, by naturality of $\phi$, it follows that:
  \begin{align*}
    \pi_1 \circ \phi_{\A{T}X}
    &= (G \circ \mathbb{T})(\Pi) \circ \phi_{\A{T}X}\\
    &= \phi_{X} \circ (F \circ \mathbb{T})(\Pi)\\
    &= f \circ \pi_1
  \end{align*}
  Hence $\phi_{\A{T}X} = \pair{f \circ \pi_1}{f'}$.
  Additionally, we also have:
  \begin{align*}
    G(\oplus) \circ \pair{f \circ \pi_1}{f'}
    &=G(\oplus) \circ \phi_{\A{T}X}\\
    &= \phi_{X} \circ F(\oplus)\\
    &= f \circ F(\oplus)
  \end{align*}
  which states precisely that $f'$ is a derivative for $f$.
\end{proof-for-lemma}

\subsection*{Stable derivatives and linearity}

\begin{proof-for-lemma}{stable-derivatives-additive}
  Since $\d f$ is regular, the following diagram commutes:
  \begin{center}
    \begin{tikzcd}[ampersand replacement=\&]
      \us A \times (\changes{A} \times \changes{A})
      \arrow[rr, "\pair{\pair{\pi_1}{\pi_1 \circ \pi_2}}{\arr{a}}"]
      \arrow[dd, "\d f \circ (\Id \times +_A)", swap]
      \&\& (\us A \times \changes{A}) \times ((\us A \times \changes{A}) \times \changes{A})
      \arrow[dd, "\d f \times (\d f \circ (\oplus_A \times \Id))"]
      \\\\
      \changes{B}
      \&\&
      \changes{B} \times \changes{B}
      \arrow[ll, "+_B"]
    \end{tikzcd}
  \end{center}
  But since $\d f$ is stable, we have $\d f \circ (\oplus \times \Id) =
  \d f \circ (\pi_1 \times \Id)$, and substituting in the previous diagram
  gives the desired result.
  
\end{proof-for-lemma}

\begin{proof-for-lemma}{composition-stable}
  \begin{align*}
    \d (\da{g} \circ \da{f}) \circ (\oplus \times \Id)
    &= \d g \circ \pair{\us{f} \circ \pi_1}{\d f} \circ (\oplus \times \Id)\\
    &= \d g \circ \pair{\us{f} \circ \pi_1 \circ (\oplus \times \Id)}{\d f \circ (\oplus \times \Id)}\\
    &= \d g \circ \pair{\us{f} \circ \oplus \pi_1}{\d f \circ (\oplus \times \Id)}\\
    &= \d g 
      \circ \pair{\oplus \circ \pair{f \circ \pi_1}{\d f} \circ \pi_1}
      {\d f \circ (\pi_1 \times \Id)}\\
    &= \d g 
      \circ (\oplus \times \Id)
      \circ \pair{\pair{\us{f} \circ \pi_1}{\d f} \circ \pi_1}{\d f \circ (\pi_1 \times \Id)}\\
    &= \d g \circ (\pi_1 \times \Id)
      \circ \pair{\pair{\us{f} \circ \pi_1}{\d f} \circ \pi_1}{\d f \circ (\pi_1 \times \Id)}\\
    &= \d g \circ \pair{\us{f} \circ \pi_1 \circ \pi_1}{\d f \circ (\pi_1 \times \Id)}\\
    &= \d g \circ \pair{\us{f} \circ \pi_1}{\d f} \circ (\pi_1 \times \Id)\\
    &= \d (\da g \circ \da f) \circ (\pi_1 \times \Id)
  \end{align*}
  
\end{proof-for-lemma}

\section{Supplementary materials for Section~\ref{sec:extrinsic}} 

\subsection*{Tangent bundles in change action models}
{}
\begin{proof-for-lemma}{tangent-bundle-monad}
  First, we verify that $\oplus \circ \ftor{T}\arr{z} = \oplus \circ \arr{z}$.
  This is easy to do since (omitting some obvious isomorphisms):
  \begin{align*}
    \ftor{T}(\arr{z}) &= \ftor{T}\pair{\Id}{0}\\
                    &= \pair{\ftor{T}\Id}{\ftor{T}0}\\
                    &= \pair{\Id}{0}\\
                    &= \arr{z}
  \end{align*}
  The equation $\oplus \circ \ftor{T}\oplus = \oplus \circ \oplus$ 
  is merely an instance of the naturality of $\oplus$ and thus it is satisfied
  trivially.
  
\end{proof-for-lemma}

\begin{proof-for-lemma}{internalisation}
  Whenever $\cat{C}$ is cartesian closed, there is a morphism $\arr{d}_{A, B} : (A \Ra B) \ra (A \times \changes{A}) \Ra \changes{B}$ such that, 
  for any morphism $f : 1 \times A \ra B$, $\arr{d}_{A, B} \circ \Lambda f = \Lambda (\d f \circ \pair{\pair{\pi_1}{\pi_{12}}}{\pair{\pi_1}{\pi_{22}}})$.

  Consider the evaluation map $\arr{ev}_{A, B} : (A \Ra B) \times A \ra B$ in $\cat{C}$.
  Its derivative $\d\arr{ev}_{A, B}$ has type
  \begin{gather*}
    \d\arr{ev}_{A, B} : ((A \Ra B) \times A) \times (\changes(A \Ra B) \times \changes A) \ra \changes B
  \end{gather*}
  (note that we make no assumptions about the structure of $\changes (A \Ra B)$).

  Then consider the following composite:
  \begin{gather*}
    \d\arr{ev}_{A, B} \circ \pair{\pair{\pi_1}{\pi_{12}}}{\pair{0_{A \Ra B}}{\pi_{22}}} : (A \Ra B) \times (A \times \changes A) \ra \changes B
  \end{gather*}

  By the universal property of the exponential, we have $\arr{ev}_{A, B} \circ (\Lambda f \times \Id_A) = f$ and therefore
  $\alpha(\arr{ev}_{A, B} \circ (\Lambda f \times \Id_A)) = \alpha(f)$. Thus:
  \begin{align*}
    \d f &= \d (\arr{ev}_{A, B} \circ (\Lambda f \times \Id_A))\\
         &= \d (\arr{ev}_{A, B} \circ \pair{\Lambda f}{\Id_A}\\
         &= \d\arr{ev}_{A, B} \circ \pair{\pair{\Lambda f}{\Id_A} \circ \pi_1}
                                        {\pair{0_{A \Ra B}}{\pi_2}}\\
         &= \d\arr{ev}_{A, B}
           \circ \pair{\pair{\pi_1}{\pi_{12}}}{\pair{0_{A \Ra B}}{\pi_{22}}} 
           \circ \pair{\Lambda f}{\pair{\pi_1}{\pi_2}}\\
         &= \d\arr{ev}_{A, B}
           \circ \pair{\pair{\pi_1}{\pi_{12}}}{\pair{0_{A \Ra B}}{\pi_{22}}} 
           \circ \pair{\Lambda f}{\Id_{A \times \changes A}}
  \end{align*}
  from which it follows trivially that
  \[
    \arr{d}_{A, B} = \Lambda(\d\arr{ev}_{A, B}
    \circ \pair{\pair{\pi_1}{\pi_{12}}}{\pair{0_{A \Ra B}}{\pi_{22}}}
  \]
  is the desired morphism.

\end{proof-for-lemma}

\begin{proof-for-theorem}{oplus-lambda-infinitesimal}
  By the Yoneda lemma, the natural transformation
  \begin{gather*}
    \oplus \circ \phi^{-1} : U \Ra A \ra A
  \end{gather*}
  is precisely evaluation at some fixed element $\mathbbm{1} : 1 \ra U$ (and,
  conversely, evaluating $\phi(t)$ at $\mathbbm{1}$ is precisely $\oplus(t)$).

  Commutativity of the above diagram then can be shown by equational reasoning
  in the internal logic of the CCC $\cat{C}$:
  \begin{align*}
    f : \ftor{T}(A \Ra B) \ts \lambda a . \oplus (\phi^{-1}(\lambda u . \phi(f)(u)(a)))
    &= \lambda a . \phi(f)(\mathbbm{1})(a)\\
    &= \phi(f)(\mathbbm{1})\\
    &= \oplus(f)
  \end{align*}
  
\end{proof-for-theorem}

\section{Supplementary materials for Section~\ref{sec:examples}}

\subsection*{Generalised cartesian differential categories}

\begin{proof-for-theorem}{GCDC-change-action-model}
  We need to check that $\alpha$ is well-defined and a right-inverse to the forgetful
  functor.
  
  First, note that $\alpha(f)$ trivially satisfies the derivative property:
  \begin{gather*}
  f \circ \pi_1 = \pi_1 \circ \pair{f \circ \pi_1}{\DF{f}}
  \end{gather*}
  Furthermore, by the axiom \ax{2} of generalised cartesian differential categories, we have:
  \begin{align*}
    \DF{f} \circ \pair{\Id}{0_A \circ !} &= 0_B\\
    \DF{f} \circ \pair{a}{+ \circ \pair{u}{v}} 
    &= + \circ \pair{\DF{f} \circ \pair{a}{u}}{\DF{f} \circ \pair{a}{v}}
  \end{align*}
  This entails that the map $\alpha(f) = (f, \DF{f})$ is indeed a differential map.
  Functoriality of $\alpha$ is a direct consequence of axioms \ax{3} and \ax{5}.
  
  Furthermore, $\alpha$ preserves products (up to isomorphism) since, by definition,
  $L(A \times B) = L(A) \times L(B)$ and by axioms \ax{3} and \ax{4}, 
  and is trivially a right-inverse to the forgetful functor. 
  Therefore $\alpha$ is a change action model.
  
\end{proof-for-theorem}

If $f : A \ra \ftor{T}B, g : B \ra \ftor{T}C$ are $\cat{C}$-morphisms, we denote their Kleisli composite $\mu \circ \ftor{T}g \circ f$ by $g \star f$. 

\begin{proof-for-lemma}{lem:GCDC-change-action-model}

  First, since $\ftor{T}$ preserves products (up to isomorphism), it follows that $\cat{C}_{\ftor{T}}$ is cartesian, with the product of objects $A, B$ in $\cat{C}_{\ftor{T}}$ being precisely the
  product $A \times B$ in $\cat{C}$. For brevity, we write $\pi_i \circ \pi_j$ as $\pi_{ij}$.

  Given $\cat{C}$-morphisms $f : A \ra \ftor{T}B$ and $g : A \ra \ftor{T}C$ in $\cat{C}$, write 
  \begin{align*}
    \phi &\defeq \pair{\pair{\pi_{11}}{\pi_{12}}}{\pair{\pi_{21}}{\pi_{22}}} : \ftor{T}A \times \ftor{T}B \ra \ftor{T}(A \times B)\\
    \kpair{f}{g} &\defeq \phi \circ \pair{f}{g} \equiv \pair{\pair{\pi_1 \circ f}{\pi_1 \circ g}}{\pair{\pi_2 \circ f}{\pi_2 \circ g}}
  \end{align*}
  with $\phi$ being the isomorphism between $\ftor{T}A \times \ftor{T}B$ and $\ftor{T}(A \times B)$ and $\kpair{f}{g}$ the universal
  morphism for the product $A \times B$ in $\cat{C}_{\ftor{T}}$.

  It is immediate that $\alpha_{\ftor{T}}$ is functorial and preserves products, and it is trivially
  a section of the forgetful functor.

  We need to prove that $\alpha_{\ftor{T}}(A)$ is a change action internal to the category $\cat{C}_{\ftor{T}}$.
  First note that since $(L_0(A), +_A, 0_A)$ is a commutative monoid in $\cat{C}$,
  the triple $(L_0(A), \eta \circ +_A, \eta \circ 0_A)$ is a commutative monoid in $\cat{C}_{\ftor{T}}$.
  
  We need to check that $\eta \circ \pi_1$ is a monoid action. First, note that $(\eta \circ g) \star f = \ftor{T}g \circ f$.
  Then:
  \begin{align*}
    \eta \circ \pi_1 \star \kpair{f}{\eta \circ 0_A}
    &= \ftor{T}\pi_1 \circ \kpair{f}{\eta \circ 0_A}\\
    &= (\pi_1 \times \pi_1) \circ \kpair{f}{\eta \circ 0)A}\\
    &= (\pi_1 \times \pi_1) \circ \pair{\pair{\pi_1 \circ f}{\pi_1 \circ \eta \circ 0}}{\pair{\pi_2 \circ f}{\pi_2 \circ \eta \circ 0}}\\
    &= \pair{\pi_1 \circ f}{\pi_2 \circ f}\\
    &= f
  \end{align*}
  That $\Id_{A \times {L_0(A)}}$ respects the associativity of $\eta \circ +_A$ follows by a similar argument.

  We write in detail the proof that the derivative condition holds. 
  In particular,
  what we seek to prove is 
  \[
  f \star \oplus = \oplus \star \kpair{f \star (\eta \circ \pi_1)}{\DF{f}}.
  \]
  Then, given that $\oplus = \eta \circ \pi_1$, and noting that $f \star \eta = f$, we obtain:
  \[
  \begin{array}{ll}
    (\eta \circ \pi_1) \star \kpair{f \star (\eta \circ \pi_1)}{\DF{f}}
    &= (\eta \circ \pi_1) \star \kpair{f \circ \pi_1}{\DF{f}}\\
    &= \ftor{T}\pi_1 \circ \kpair{f \circ \pi_1}{\DF{f}}\\
    &= (\pi_1 \times \pi_1) \circ \kpair{f \circ \pi_1}{\DF{f}}\\
    &= f \circ \pi_1
  \end{array}
  \]
  On the other hand:
  \begin{align*}
    f \star \oplus
    & = f \star (\eta \circ \pi_1)\\
    & = (f \star \eta) \circ \pi_1\\
    & = f \circ \pi_1
  \end{align*}
  Hence the derivative condition trivially holds.
  
  Regularity is a straightforward consequence of the additivity of derivatives in generalised cartesian differential categories.
  
\end{proof-for-lemma}

\subsubsection{The Eilenberg-Moore model}
\label{sec:em-model}

\begin{defi}\rm
  Given a category $\cat{C}$ and a monad $(\ftor{T}, \eta, \mu)$, a \emph{$\ftor{T}$-algebra} is
  a pair $(A, \nu)$ where $A$ is an object in $\cat{C}$ and $\nu : \ftor{T}A \ra A$ is a 
  $\cat{C}$-morphism such that:
  \begin{itemize}
  \item $\nu \circ \ftor{T}\nu = \nu \circ \mu$ 
  \item $\nu \circ \eta = \Id$
  \end{itemize}
  
  A morphism of $\ftor{T}$-algebras between $\ftor{T}$-algebras $(A, \nu_A), (B, \nu_B)$ is
  a $\cat{C}$-morphism $f : A \ra B$ such that $f \circ \nu_A = \nu_B \circ \ftor{T}f$.
\end{defi}

Both of the previous change action models are in fact categories of algebras for the tangent
bundle monad $\ftor{T}$ on a generalised cartesian differential category $\cat{C}$ (the flat model
considers algebras of the form $(A, \pi_1)$, whereas the Kleisli category for $\ftor{T}$ can be
understood as the category of freely generated $\ftor{T}$-algebras). This is a consequence of the
following result:

\begin{thm}
\label{thm:t-algebras-are-change-actions}
  Let $(A, \nu_A)$ be a $\ftor{T}$-algebra such that $\DF{\nu_A} \circ \pair{f}{\pair{g}{0_A}} = g$.
  Then the tuple $(A, L_0(A), \nu_A, +_A, 0_A)$ is a change 
  action. Furthermore, given a $\cat{C}$-morphism $f : A \ra B$, it is a $\ftor{T}$-algebra morphism
  between $(A, \nu_A)$ and $(B, \nu_B)$ if and only if $\ftor{T}f$ is a differential morphism between 
  the corresponding change actions.
\end{thm}

\begin{proof}
  For the first part, it suffices to check that $\nu_A$ is a monoid action.
  
  First, note that the monad unit $\epsilon_A$ is precisely the map $\pair{\Id_A}{0_A}$, hence since
  $\nu_A$ is a $\ftor{T}$-algebra morphism we have $\nu_A \circ \pair{\Id_A}{0_A} = \Id_A$.

  For the second part, note that the monoid addition can be written in terms of the monad 
  multiplication as follows:
  \begin{align*}
    \mu \circ \pair{\pair{\pi_1}{\pi_{12}}}{\pair{\pi_{22}}{0_A}}
    &=
      (\Id \times +_A)
      \circ \pair{\pi_{11}}{\pair{\pi_{21}}{\pi_{12}}}
      \circ \pair{\pair{\pi_1}{\pi_{12}}}{\pair{\pi_{22}}{0_A}}\\
    &= (\Id \times +_A)
      \circ \pair{\pi_1}{\pair{\pi_{12}}{\pi_{22}}}\\
    &= (\Id \times +_A)
      \circ \Id\\
    &= (\Id \times +_A)
  \end{align*}
  Then because $\nu_A$ is a $\ftor{T}$-algebra homomorphism, we have:
  \begin{align*}
    \nu_A \circ (\Id \times +_A)
    &= \nu_A \circ \mu \circ \pair{\pair{\pi_1}{\pi_{12}}}{\pair{\pi_{22}}{0_A}}\\
    &= \nu_A \circ \ftor{T}\nu_A \circ \pair{\pair{\pi_1}{\pi_{12}}}{\pair{\pi_{22}}{0_A}}\\
    &= \nu_A 
      \circ \pair{\nu_A \circ \pi_1}{\DF{\nu_A}} 
      \circ \pair{\pair{\pi_1}{\pi_{12}}}{\pair{\pi_{22}}{0_A}}\\
    &= \nu_A \circ 
      \pair{\nu_A \circ \pair{\pi_1}{\pi_{12}}}
           {\DF{\nu_A} \circ \pair{\pair{\pi_1}{\pi_{12}}}{\pair{\pi_{22}}{0_A}}}\\
    &= \nu_A \circ 
      \pair{\nu_A \circ \pair{\pi_1}{\pi_{12}}}
           {\DF{\nu_A} \circ \pair{\pair{\pi_1}{\pi_{12}}}{\pair{\pi_{22}}{0_A}}}\\
    &= \nu_A \circ
      \pair{\nu_A \circ \pair{\pi_1}{\pi_{12}}}
           {\pi_{22}}\\
    &= \nu_A \circ (\nu_A \times \Id) \circ \pair{\pair{\pi_1}{\pi_{12}}}{\pi_{22}}
  \end{align*}

  \map{
    It's extremely unsatisfying that I need to require the condition on $\DF{\nu}$ to get this, but I
    see no other way to prove it. I have some intuitive ideas of why this is.
  }
  
  Now consider a $\cat{C}$-morphism $f : A \ra B$. Its derivative $\DF{f}$ satisfies regularity
  trivially, since in every generalised cartesian differential category derivatives are additive in
  their second argument. Then the property that $f$ is a $\ftor{T}$-algebra morphism between
  $(A, \nu_A)$ and $(B, \nu_B)$ states that $\nu_B \circ \ftor{T}f = f \circ \nu_A$, which
  is equivalent to stating that $(f, \DF{f})$ is a differential morphism between the corresponding
  change actions.
  
\end{proof}

One might be tempted to generalise this result in the ``obvious'' direction and try to construct a change action model directly on the Eilenberg-Moore category $\cat{C}^{\ftor{T}}$. This is, however, not possible:
while an algebra on some object $A$ does define a change action, it does not give any information on 
what should be the change action structure on $\Delta A$.

Instead, consider the obvious projection functor $\pi:\cat{C}^{\ftor{T}} \ra \cat{C}$.
Now let $\sigma : \cat{C} \ra \cat{C}^{\ftor{T}}$ be a section of $\pi$, i.e. $\pi \circ \sigma =
\ftor{Id}$. For every object $A$ of $\cat{C}$, the section $\sigma$ picks a particular 
$\ftor{T}$-algebra $\sigma(A) = (A, \nu_A)$. Similarly, $\sigma$ maps every morphism $f : A \ra B$ in 
$\cat{C}$ onto a $\ftor{T}$-algebra homomorphism.

An immediate corollary of Theorem~\ref{thm:t-algebras-are-change-actions} is that any section of
$\pi$ that maps objects of $\cat{C}$ to ``well-behaved'' $\ftor{T}$-algebras defines a model
structure on $\cat{C}$.

\begin{lem}
  Let $\sigma : \cat{C} \ra \cat{C}^{\ftor{T}}$ be a section of $\pi : \cat{C}^{\ftor{T}} \ra \cat{C}$
  such that for every object $A$ of $\cat{C}$, the corresponding algebra $(A, \nu_A)$ satisfies
  $\DF{\nu_A} \circ \pair{f}{\pair{g}{0_A}} = g$. Then there is a change action model $\alpha : \cat{C}
  \ra \ftor{\CAct}(\cat{C})$ defined by:
  \begin{itemize}
    \item $\alpha(A) = (A, L_0(A), \nu_A, +_A, 0_A)$
    \item $\alpha(f) = \ftor{T}f = \pair{f \circ \pi_1}{\DF{f}}$
  \end{itemize}

\end{lem}

The flat model described in Section~\ref{sec:flat-model} is an immediate corollary of this, as it
is the change action model obtained from picking the section that maps every object $A$ to the
$\ftor{T}$-algebra $\pi_2 : A \times L_0(A) \ra A$.

\map{Moved these so they follow the structure of the body}
\subsection*{Groups, calculus of differences, and Boolean differential calculus}

\begin{proof-for-lemma}{lem:boolean-derivative}
Note first that in any Boolean algebra $\BB$, we have $\neg u = u \nlra \top$. 
Moreover
\begin{gather*}
  (u_1, \ldots, \neg u_i, \ldots, u_n) = (u_1, \ldots, u_n) \oplus (\bot, \ldots, \top, \ldots, \bot)
\end{gather*}
Furthermore:
\[
\begin{array}{ll}
  &\displaystyle f(u_1, \ldots, u_n) \oplus \frac{\d f}{\d x_i}(u_1, \ldots, u_n)\\
  =& f(u_1, \ldots, u_n) \oplus (f(u_1, \ldots, u_n) \nlra f(u_1, \ldots, \neg u_i, \ldots, u_n))\\
  =& f(u_1, \ldots, u_n) \nlra (f(u_1, \ldots, u_n) \nlra f(u_1, \ldots, \neg u_i, \ldots, u_n))\\
  =& (f(u_1, \ldots, u_n) \nlra f(u_1, \ldots, u_n)) \nlra f(u_1, \ldots, \neg u_i, \ldots, u_n)\\
  =& \bot \nlra f(u_1, \ldots, \neg u_i, \ldots, u_n)\\
  =& f(u_1, \ldots, \neg u_i, \ldots, u_n)\\
  =& f((u_1, \ldots, u_n) \nlra \top_i)\\
  =& f((u_1, \ldots, u_n) \oplus \top_i)
\end{array}
\]
Thus, since derivatives in $\CGrpS$ are unique, the Boolean derivative 
\[
\displaystyle \frac{\d f}{\d x_i}(u_1, \ldots, u_n)
\] 
is precisely the
derivative $\d f((u_1, \ldots, u_n), \top_i)$.

\end{proof-for-lemma}

\subsection*{Polynomials over commutative Kleene algebras}

\begin{proof-for-lemma}{kleene-change-actions-model}

We consider the essential case of $m = n = 1$; the proof of the lemma is then a straightforward generalisation.

We shall make use of the following properties of \emph{commutative} Kleene algebras.
\begin{enumerate}
\item $(a_1 + \cdots + a_m)^n = \sum \{a_1^{i_1} \cdots a_m^{i_m} \mid i_1 + \cdots + i_m = n; i_1, \cdots, i_m \geq 0\}$.

Since $(a + b)^\star = a^\star \, b^\star$, we have
\[
\big((a_1 + \cdots + a_m)^n\big)^\star = \prod \{(a_1^{i_1} \cdots a_m^{i_m})^\star \mid i_1 + \cdots + i_m = n; i_1, \cdots, i_m \geq 0\}.
\] 
For example $((a + b + c)^2)^\star = (a \, a)^\star \, (a \, b)^\star \, (a \, c)^\star \, (b \, b)^\star \, (b \, c)^\star \, (c \, c)^\star$.

\item Pilling's Normal Form Theorem \cite{Pilling73,DBLP:conf/lics/HopkinsK99}: every (regular) expression is equivalent to a sum $y_1 + \cdots + y_n$ where each $y_i$ is a product of atomic symbols and expressions of the form $(a_1 \cdots a_k)^\star$, where the $a_i$ are atomic symbols. 
For example $(((a \, b)^\star c)^\star + d)^\star = d^\star + (a \, b)^\star c^\star c \, d^\star$.
\end{enumerate}

Take $p(x) \in \KK[x]$, viewed as a function from change action $(\KK, \KK, +, +, 0)$ to itself.
For $a, b \in \KK$, we have 
\[
\d p(a, b) \defeq \der{p}{x}(a + b) \cdot b.
\]
That this defines a derivative of $p(x)$ is an immediate consequence of Theorem~\ref{thm:taylor}.

We need to prove that the derivative is regular.
Trivially $\d p (a, 0) = 0$.
It remains to prove: for $u, a, b \in \KK$
\begin{equation}
\der{p}{x}(u + a + b) \cdot (a + b) =
\der{p}{x}(u + a) \cdot a + \der{p}{x}(u + a + b) \cdot b
\label{eq:kleene-regular}
\end{equation}
which we argue by structural induction,
presenting the cases of $p = q^\star$ and $p = q \, r$ explicitly.

Let $p = q^\star$. 
Thanks to Pilling's Normal Form Theorem, WLOG we assume $q = x^{n+1} \, c$.
Now $\displaystyle \der{x^{n+1} \, c}{x}(x) = x^n \, c$.
Then $\displaystyle \der{p}{x}(x) = q^\star (x) \, \der{q}{x}(x) = (x^{n+1} \, c)^\star (x^n \, c)$.
Clearly $\mathrm{RHS}(\ref{eq:kleene-regular}) \leq \mathrm{LHS}(\ref{eq:kleene-regular})$.
For the opposite containment, it suffices to show 
\[
\der{p}{x}(u + a + b) \cdot a \ \leq \ \der{p}{x}(u + a) \cdot a + \der{p}{x}(u + a + b) \cdot b
\]
I.e.
\begin{equation}
(\theta^{n+1} \, c)^\star \, (\theta^n \, c) \, a 
\ \leq \
((u+a)^{n+1} \, c)^\star \, ((u+a)^n \, c) \, a + (\theta^{n+1} \, c)^\star \, (\theta^n \, c) \, b
\label{eq:kleene-regular-2}
\end{equation} 
using the shorthand $\theta = u + a + b$.

A typical element that matches LHS(\ref{eq:kleene-regular-2}) has shape 
\[
\Xi \defeq (u^{i'} \, a^{j'} \, b^{k'} \, c)^l \, (u^i \, a^j \, b^k \, c) \, a
\] 
satisfying 
\[
l \geq 0, \quad i' + j' + k' = n+1, \quad i + j +k = n.
\]
It suffices to consider two cases: $l = 0$ and $l = 1$, for if $l > 1$ and $(u^{i'} \, a^{j'} \, b^{k'} \, c) \, (u^i \, a^j \, b^k \, c) \, a$ matches RHS(\ref{eq:kleene-regular-2}) then so does $(u^{i'} \, a^{j'} \, b^{k'} \, c)^l \, (u^i \, a^j \, b^k \, c) \, a$. 
\begin{itemize}
\item Now suppose $l = 0$. 
If $k = 0$ then $\Xi$ matches the first summand of RHS(\ref{eq:kleene-regular-2});
otherwise note that $\Xi = (u^i \, a^{j+1} \, b^{k-1} \, c) \, b$ matches the second summand of RHS(\ref{eq:kleene-regular-2}).
\item Next suppose $l = 1$.
If $k = k' = 0$ then $\Xi$ matches the first summand of RHS(\ref{eq:kleene-regular-2});
otherwise suppose $k' > 0$ then $\Xi = (u^{i'} \, a^{j'+1} \, b^{k'-1} \, c) \, (u^i \, a^j \, b^k \, c) \, b$ matches the second summand of RHS(\ref{eq:kleene-regular-2}).
\end{itemize}

Let $p (x) = q (x) \, r (x)$.
Applying the product rule of partial derivatives, equation (\ref{eq:kleene-regular}) is equivalent to $L = R$ where
\[
\begin{array}{rll}
L &\defeq & \displaystyle
\left[
r(\theta) \, \der{q}{x} (\theta) + 
q(\theta) \, \der{r}{x} (\theta)
\right] \cdot (a + b)
\\
R & \defeq & \displaystyle
\left[
r(u + a) \, \der{q}{x} (u + a) + 
q(u + a) \, \der{r}{x} (u + a)
\right] \cdot a \\
& & \displaystyle + \ \left[
r(\theta) \, \der{q}{x} (\theta) + 
q(\theta) \, \der{r}{x} (\theta)
\right] \cdot b
\end{array}
\]
using the shorthand $\theta = u + a + b$ as before.
Similar to the preceding case, clearly $R \leq L$.
To show $L \leq R$, it suffices to show:
\begin{align*}
r(\theta) \, \der{q}{x} (\theta) \, a &\leq R\\ 
q(\theta) \, \der{r}{x} (\theta) \, a &\leq R
\end{align*}
We consider the first; the same reasoning applies to the second.
As before, thanks to Pilling's Normal Form Theorem, we may assume that $r(x) = (x^{m+1} \, c)^\star$ and $q(x) = (x^{n+1} \, d)^\star$.
Then 
$\displaystyle r(\theta) \, \der{q}{x} (\theta) \, a = (\theta^{m+1} \, c)^\star \, (\theta^{n+1} \, d)^\star \, (\theta^{n} \, d) \, a$.
By considering a typical element $\Xi$ that matches the preceding expression, 
and using the same reasoning as the preceding case, we can then show that $\Xi$ matches $R$, as desired. 

\end{proof-for-lemma}

\section{Supplementary materials for Section~\ref{sec:omega-change-actions}}

\noindent\emph{Notation}. Let $\alpha$ be an $\omega$-sequence. 
We use shorthand $(\alpha)_j = \pr_j \alpha$.

\begin{proof-for-prop}{id-law}
  To see that $\Id \circ \seq{f_i} = \seq{f_i}$, 
  we show
  \[
    (\Id \circ \seq{f_i})_{n}
    =
    \cpow{n} \circ \pr_0 \A{D}^n \seq{f_i}
    =
    f_n = (\seq{f_i})_n
  \]
  by a straightforward induction on $n$.

  For the other direction of the identity law, $\seq{g_i} \circ \Id = \seq{g_i}$, 
  it suffices to prove: for all $n \geq 0$
  \begin{equation}
    \pr_0 \A{D}^n \Id \ = \ \Id : \pr_{n} \A{D} \changeseq A \to \pr_{n} \A{D} \changeseq A
    \label{eq:id-law}
  \end{equation}
  For instance we have
  \begin{align*}
    \pr_0 \A{D} \Id&= \pair{\pi_1}{\pi_2} : \pr_1 \A{D} \changeseq A \to  \pr_1 \A{D} \changeseq A\\
    \pr_0 \A{D}^2 \Id &= \pair{\pair{\pi_1}{\pi_2} \circ \pi_1}{\pair{\pi_2 \circ \epow{1}}{\cpow{2}}}
                           : \pr_2 \A{D} \changeseq A \to  \pr_2 \A{D} \changeseq A
    \\
    \pr_0 \A{D}^3 \Id
                      &= 
                        \pair{
                        \pi_1
                        }{
                        \pair{
                        \pair{\pi_2 \circ \epow{1}}{\cpow{2}} \circ \epow{1}
                        }{
                        \pair{\cpow{2} \circ \epow{2}}{\cpow{3}}
                        }  
                        }  :
                        \pr_3 \A{D} \changeseq A \to  \pr_3 \A{D} \changeseq A
  \end{align*}

  To establish (\ref{eq:id-law}), we need to prove a stronger result: 
  \begin{lem}
    For all $n, j \geq 0$,
    \(
    \pr_j \A{D}^n \Id = \cpow{j} : \pr_{n+j} \A{D} \changeseq A \to \pr_{n} \A{D} \Pi^j \changeseq A.
    \)
  \end{lem}

  \begin{proof}
    We use lexicographical induction on $(n, j)$.
    The base case is straightforward.
    Our induction hypothesis is
    \[
      \forall j \geq 0 \, . \, 
      \pr_j\A{D}^n\Id \ = \ \cpow{j} : \pr_{n+j} \A{D} \changeseq A \to \pr_{n} \A{D} \Pi^j \changeseq A
    \]
    Then
    \[\begin{array}{ll}
        & \pr_{j} \A{D}^{n+1}\Id \\
        = & \pair{
            \pr_j \A{D}^n \Id \circ \epow{j}
            }{
            \pr_{j+1} \A{D}^n \Id
            } \\
        = & \pair{
            \cpow{j} \circ \epow{j}
            }{
            \cpow{j+1}
            } \ : \ \pr_{n+1 + j}\A{D}\seq{A_i} \ra 
            \pr_{n} \A{D} \Pi^{j} \changeseq A
            \times 
            \pr_{n} \A{D} \Pi^{j+1} \changeseq A\\
        = & \cpow{j} \ : \ \pr_{n+1 + j}\A{D}\seq{A_i} \ra 
            \pr_{n+1} \A{D} \Pi^{j} \changeseq A
      \end{array}
    \]
    as desired.
    The third equality uses the fact:
    for $j \geq 0$, $\cpow{j} \circ \epow{j} = \pi_1 \circ \cpow{j}$, 
    which is proved by a straightforward induction on $j$.
    The base case is trivial.
    For the inductive case: 
    \[
    \cpow{j+1} \circ \epow{j+1} = \cpow{j+1} \circ \pair{\epow{j}}{\epow{j}} = \cpow{j} \circ \epow{j} \circ \pi_2 = \pi_1 \circ \cpow{j+1}.
    \]
    
  \end{proof}
\end{proof-for-prop}

\begin{proof-for-prop}{comp-assoc}
    For convenience and to save space, we write $\pr_n \seq{f_i}$ as $f_n$.
    Next we first prove a couple of useful technical lemmas.

    First observe that for each $n \geq 0$, and for each $0 \leq j \leq n$, there exists a morphism:
\[
  \epow{j} : \pr_{n+1} \A{D} \changeseq A \to \pr_n \A{D} \changeseq A
\]
This leads to the following.
\begin{lem}
\label{lem:assoc-nj1}
For all $n \geq 0$ and $1 \leq j \leq n$,
\[
  (\A{D}^{n} \seq{f_i})_0 \circ \epow{j} = 
  \epow{j} \circ (\A{D}^{n+1} \seq{f_i})_0 \ : \ \pr_{n+1} \A{D} \changeseq A \to \pr_{n} \A{D} \changeseq B
\]
\end{lem}

\noindent \emph{Notation}: We use the shorthand $\tmap{f}{n}{i} = \pr_i(\A{D}^n \seq{f_i})$. 

\begin{proof}
A stronger induction principle is needed.
We claim: for all $n, i \geq 0$ and $j \leq n$
\[\tmap{f}{n}{i} \circ \epow{i+j} =
\epow{j} \circ \tmap{f}{n+1}{i}\]
We prove by lexicographical induction on $(n, j, i)$.
The base case, which is 
\(\forall i \geq 0 \, . \,
f_i \circ \epow{i} = \pi_1 \circ \tmap{f}{1}{i}
\),
holds trivially.
For the inductive case:
\[
\begin{array}{ll}
& \tmap{f}{n+1}{i} \circ \epow{i+j+1}\\
= &
\pair{\tmap{f}{n}{i} \circ \epow{i} \circ \epow{i+j+1}}{\tmap{f}{n}{i+1} \circ \epow{i+j+1}}
\\
= &
\pair{\tmap{f}{n}{i} \circ \epow{i+j} \circ \epow{i}}{\tmap{f}{n}{i+1} \circ \epow{i+j+1}}
\\
= &
\pair{\epow{j} \circ \tmap{f}{n}{i} \circ \epow{i} \circ \epow{i}}{\epow{j} \circ \tmap{f}{n}{i+1}}
\\
= & 
\epow{j} \times \epow{j} \circ \pair{\tmap{f}{n+1}{i} \circ \epow{i}}{\tmap{f}{n+1}{i+1}}
\\
= &
\epow{j+1} \circ \tmap{f}{n+2}{i}
\end{array}
\]
The second equality appeals to the fact:
for $i, j \geq 0$, $n > i+j+1$
\[
\epow{i} \circ \epow{i+j+1} = \epow{i+j} \circ \epow{i} : \pr_{n} \A{D} \changeseq A \to \pr_{n-2} \A{D} \changeseq A
\]
which is easily proved by induction.

\end{proof}

\begin{lem}
\label{lem:assoc-nj2}
For all $n\geq 0$ and $0 \leq j \leq n$, 
\[
\big(\A{D}^n (\seq{g_i} \circ \seq{f_i})\big)_j = (\A{D}^n \seq{g_i})_j \circ (\A{D}^{n+j} \seq{f_i})_0 \ : \ \pr_{n+j} \A{D} \changeseq A \to \pr_{n} \A{D} \Pi^j \changeseq C.
\]
\end{lem}

\begin{proof}
We prove by lexicographical induction on $(n, j)$.
The base case is straightforward. 
For the inductive case:
\[
\begin{array}{ll}
& \big(\A{D}^{n+1} (\seq{g_i} \circ \seq{f_i})\big)_{j}
\\
= & 
\pair{
  \big(\A{D}^{n} (\seq{g_i} \circ \seq{f_i})\big)_j
  \circ \epow{j}
}{
  \big(\A{D}^n (\seq{g_i} \circ \seq{f_i})\big)_{j+1}
}
\\
= & 
\pair{
  (\A{D}^{n} \seq{g_i})_j \circ (\A{D}^{n+j}\seq{f_i})_0
  \circ \epow{j}
}{
  \big(\A{D}^n (\seq{g_i} \circ \seq{f_i})\big)_{j+1}
}
\\
= & 
\pair{
  (\A{D}^{n} \seq{g_i})_j \circ \epow{j}\circ (\A{D}^{n+j+1}\seq{f_i})_0
}{
  (\A{D}^n \seq{g_i})_{j+1} \circ (\A{D}^{n+j+1}\seq{f_i})_{0}
}
\\
= & 
\pair{
  (\A{D}^{n} \seq{g_i})_j \circ \epow{j}}{
  (\A{D}^n \seq{g_i})_{j+1}
}
\circ (\A{D}^{n+j+1}\seq{f_i})_{0}
\\
= & 
(\A{D}^{n+1} \seq{g_i})_j
\circ (\A{D}^{n+j+1}\seq{f_i})_{0}
\end{array}
\]
The second equality appeals to the induction hypothesis.
The third equality uses Lemma~\ref{lem:assoc-nj1} for the first component, and the induction hypothesis for the second.

\end{proof}

We are now ready to prove the  associativity of composition, which boils down to: for all $n \geq 0$, 
\[
\big(\A{D}^n (\seq{g_i} \circ \seq{f_i})\big)_0 \ = \ (\A{D}^n \seq{g_i})_0 \circ (\A{D}^n \seq{f_i})_0.
\]
which is a special case of Lemma~\ref{lem:assoc-nj2}.
\end{proof-for-prop}

\begin{lem}
  Whenever $\ds{f} : \ds{A} \ra \ds{B}$ is an $\omega$-differential map, then so is $\A{D}^n \seq{f_i}$ for all $n \geq 0$.
\end{lem}

\begin{proof}
  Since $f_{i + 1}$ is always a regular derivative for $f_i$, it suffices to show that $f_{i + 1} \circ \epow{i + 1}$ is a regular derivative for $f_i \circ \epow{i}$.

  We abuse the notation and write $\d f$ to denote some arbitrary derivative of $f$. 
  Then we show by induction on $i$ that $\epow{i} \circ \pi_2$ is a regular derivative of $\epow{i}$. The base case is trivial. 
  Then, by applying properties of the product of change actions (e.g. that the derivative of $\pi_i$ is $\pi_i \circ \pi_2$), we obtain:
  \begin{align*}
    \d (\epow{i + 1}) &= \d (\epow{i} \times \epow{i})\\
                        &= \d (\pair{\epow{i} \circ \pi_1}{\epow{i} \circ \pi_2})\\
                        &= \pair{\d (\epow{i} \circ \pi_1)}{\d (\epow{i} \circ \pi_2)}\\
                        &= \pair{\epow{i} \circ \pi_2 \circ \pair{\pi_1 \circ \pi_1}{\pi_1 \circ \pi_2}}{\epow{i} \circ \pi_2 \circ \pair{\pi_2 \circ \pi_1}{\pi_2 \circ \pi_2}}\\
                        &= \pair{\epow{i} \circ \pi_1 \circ \pi_2}{\epow{i} \circ \pi_2 \circ \pi_2}\\
                        &= \pair{\epow{i} \circ \pi_1}{\epow{i} \circ \pi_2} \circ \pi_2\\
                        &= (\epow{i} \times \epow{i}) \circ \pi_2\\
                        &= \epow{i + 1} \circ \pi_2
  \end{align*}
  Hence $\epow{i+1} \circ \pi_2$ is a derivative of $\epow{i+1}$. By the induction hypothesis, $\epow{i} \circ \pi_2$ is a regular derivative for $\epow{i}$, and since we have obtained a derivative for
  $\epow{i+1}$ by applying the chain rule to a composition of functions with regular derivatives, the resulting derivative is also regular.
  
  The desired result then follows from applying the chain rule again:
  \begin{align*}
    \d (f_i \circ \epow{i}) &= \d f_i \circ \pair{\epow{i} \circ \pi_1}{\d \epow{i}}\\
                            &= f_{i + 1} \circ \pair{\epow{i} \circ \pi_1}{\epow{i} \circ \pi_2}\\
                            &= f_{i + 1} \circ (\epow{i} \times \epow{i})\\
                            &= f_{i + 1} \circ \epow{i + 1}
  \end{align*}
  Hence $f_{i+1} \circ \epow{i+1}$ is a derivative for $f_i \circ \epow{i}$, and since regularity is preserved under applications of the chain rule, it is also a regular derivative.
  
\end{proof}

\begin{proof-for-lemma}{composite-omega-differential}
  As in the previous proof, we write $\d f$ for an arbitrary derivative of $f$. 
  By the previous lemma, the map $\A{D}^j\seq{f_i}$ is $\omega$-differential and, in particular, $\d \pr_0\A{D}^j\seq{f_i} = \pr_1\A{D}^j\seq{f_i}$. 
  Then, by applying the chain rule in the definition of composition, we obtain:
  \begin{align*}
    \d \pr_{j} (\seq{g_i} \circ \seq{f_i}) 
    &= \d (g_j \circ \pr_0\A{D}^j\seq{f_i})\\
    &= g_{j+1} \circ \pair{\pr_0\A{D}^j\seq{f_i} \circ \pi_1}{\d \pr_0\A{D}^j\seq{f_i}}\\
    &= g_{j+1} \circ \pair{\pr_0\A{D}^j\seq{f_i} \circ \pi_1}{\pr_1\A{D}^j\seq{f_i}}\\
    &= g_{j+1} \circ \pr_0\A{D}^{j+1}\seq{f_i}\\
    &= \pr_{j+1} (\seq{g_i} \circ \seq{f_i})
  \end{align*}

\end{proof-for-lemma}

\begin{proof-for-lemma}{id-omega-differential}
We show: for any $\omega$-differential maps $\ds{f} : \ds{A} \ra \ds{B}$ and $g : \ds{B} \ra \ds{A}$, 
  \(
    \ds{f} \circ \ds{\Id} = \ds{f}
  \) and
  \(    
  \ds{\Id} \circ \ds{g} = \ds{g}.
  \)

  $\pr_1 \ds{\Id}$ is trivially a derivative for $\pr_0 \ds{\Id}$. Furthermore, since
  $\ds{\pi_2}$ is an $\omega$-differential map, we have $\pr_{i + 2}\ds{\Id}$ is a derivative
  for $\pr_{i+1}\ds{\Id}$, therefore $\Id$ is $\omega$-differential.
  
  We prove identity by simultaneous induction:
  \begin{align*}
    \pr_0 (\ds{f} \circ \ds{\Id}) &= \pr_0 \ds{f} \circ \pr_0 \ds{\Id}\\
                                     &= \pr_0 \ds{f} \circ \Id\\
                                     &= \pr_0 \ds{f}\\
    \pr_0 (\ds{\Id} \circ \ds{g}) &= \pr_0 \ds{\Id} \circ \pr_0 \ds{g}\\
                                     &= \Id \circ \pr_0 \ds{g}\\
                                     &= \pr_0 \ds{g}
  \end{align*}
  \begin{align*}
    \pr_{i + 1} (\ds{f} \circ \ds{\Id}) &= \pr_i (\Pi \ds{f} \circ \pair{\ds{\Id} \circ \ds{\pi_1}}{\Pi \ds{\Id}})\\
                                           &= \pr_i (\Pi \ds{f} \circ \pair{\ds{\pi_1}}{\ds{\pi_2}})\\
                                           &= \pr_i (\Pi \ds{f})\\
                                           &= \pr_{i+1} \ds{f}\\
    \pr_{i + 1} (\ds{\Id} \circ \ds{g}) &= \pr_i (\Pi \ds{\Id} \circ \pair{\ds{g} \circ \ds{\pi_1}}{\Pi\ds{g}})\\
                                           &= \pr_i (\A{\pi_2} \circ \pair{\ds{g} \circ \ds{\pi_1}}{\Pi\ds{g}})\\
                                           &= \pr_i (\Pi\ds{g})\\
                                           &= \pr_{i+1} \ds{g}
  \end{align*}
  
\end{proof-for-lemma}

\begin{proof-for-lemma}{cact-omega-products}
  First we prove that the map $\pi_{\mathbf 1} : \ds{A} \times \ds{B} \ra \ds{A}$ is $\omega$-differential.
  For this we need to check that $\pr_{i+1}\pi_{\mathbf 1}$ is a regular derivative for $\pr_i\pi_{\mathbf 1}$ as a map
  from the change action $\A{D}(\ds{A} \times \ds{B}, i)$ to the change action $\changes(\ds{A}, i)$.
  
  The key insight is that $\A{D}(\ds{A} \times \ds{B}, i) = \A{D}(\ds{A}, i) \times \A{D}(\ds{B}, i)$.
  \map{Trivial to prove with induction, should I write it?} \lo{Yes we should, in the arXiv version.}
  Verifying the derivative property then boils down to applying the structure of products of change actions.

  Consider the map $\pi_1 \circ \cpow{i}$ from the carrier $\us{\A{D}(\ds{A}, i) \times \A{D}(\ds{B}, i)}$ into $\us{\changes(\ds{A}, i)}$.
  This is in fact the composition of $i+1$ differentiable maps, and hence we can apply the chain rule to compute a regular derivative for it in terms of the
  (regular) derivatives for $\pi_1, \pi_2$, which we know are $\pi_{21}, \pi_{22}$ respectively.
  We abuse the notation and write $\d f$ to denote some arbitrary derivative of $f$. 
  Then we show by induction on $i$ that $\pi_1 \circ \cpow{i + 1}$ is a derivative for $\pi_1 \circ \cpow{i}$. The base case is trivial. For the inductive case:
  \begin{align*}
    \d (\pi_1 \circ \cpow{i + 1}) 
    &= \d (\pi_1 \circ \cpow{i} \circ \pi_2)\\
    &= \d (\pi_1 \circ \cpow{i}) \circ \pair{\pi_2 \circ \pi_1}{\d \pi_2}\\
    &= \pi_1 \circ \cpow{i+1} \circ \pair{\pi_2 \circ \pi_1}{\pi_2 \circ \pi_2}\\
    &= \pi_1 \circ \cpow{i} \circ \pi_2 \circ \pair{\pi_2 \circ \pi_1}{\pi_2 \circ \pi_2}\\
    &= \pi_1 \circ \cpow{i} \circ \pi_2 \circ \pi_2\\
    &= \pi_1 \circ \cpow{i+2}
  \end{align*}
  Hence $\ds{\pi_{\mathbf 1}}$ is an $\omega$-differential map; similarly for $\ds{\pi_{\mathbf 2}}$.
  A similar argument also shows that $\pair{\ds{f}}{\ds{g}}$ is $\omega$-differential whenever $\ds{f}$ and $\ds{g}$ are.
  \map{Not sure whether to expand on that last point.}
  \lo{We should do it for the arXiv version.}

  We then check that $\ds{\pi_{\mathbf 1}} \circ \pair{\ds{f}}{\ds{g}} = \ds{f}$.
  For this, the following auxiliary lemma will be of use:
  \begin{lem}
    \label{lem:cpow-composition}
    For any 
    \lochanged{pre-$\omega$-differential map} $\seq{f_i}$, for all $j, k \geq 0$, 
    \[
    \cpow{j} \circ \pr_k \A{D}^j\seq{f_i} = f_{j + k}.
    \]
  \end{lem}

  \begin{proof}
    We proceed by induction on $j$. 
    For the case $j = 0$ we have
    $\cpow{0} \circ \pr_k \A{D}^0\seq{f_i} = \pr_k \seq{f_i}$

    For the inductive case:
    \begin{align*}
      \cpow{j+1} \circ \pr_k \A{D}^{j+1}\seq{f_i}
      &= \cpow{j} \circ \pi_2 \circ \pair{\ldots}{\pr_{k+1}\A{D}^j\seq{f_i}}\\
      &= \cpow{j} \circ \pr_{k+1}\A{D}^j\seq{f_i}\\
      &= f_{(j + 1) + k}
    \end{align*}
  \end{proof}
  
  The desired equation follows as a trivial corollary. Indeed, for any $i$
  \begin{align*}
    \pr_i (\ds{\pi_{\mathbf 1}} \circ \pair{\ds f}{\ds g})
    &= \pr_i\ds{\pi_{\mathbf 1}} \circ \pr_0 \A{D}^i\pair{\ds{f}}{\ds{g}}\\
    &= \pi_1 \circ \cpow{i} \circ \pr_0\A{D}^i\pair{\ds{f}}{\ds{g}}\\
    &= \pi_1 \circ \pr_i\pair{\ds{f}}{\ds{g}}\\
    &= \pr_i\ds{f}
  \end{align*}
  Hence $\ds{\pi_{\mathbf 1}} \circ \pair{\ds f}{\ds g} = \ds{f}$.
  
\end{proof-for-lemma}

\begin{proof-for-theorem}{cact-omega-ccc}
  Consider $\omega$-change actions $\ds{A}, \ds{B}, \ds{C}$ and let $\ds{f} : \ds{A} \ra \ds{B},
  \ds{g} : \ds{A} \ra \ds{C}$ be $\omega$-differential maps. 
  Lemma~\ref{lem:cact-omega-products} already shows that 
  $\ds{\pi_{\mathbf 1}} \circ \pair{\ds f}{\ds g} = \ds{f}$, and similarly for $\ds{\pi_{\mathbf 2}}$. 
  It remains to establish uniqueness. 

  Suppose there is an $\omega$-differential map $\ds h : \ds A \ra \ds B \times \ds C$ satisfying
  \[
    \ds{\pi_{\mathbf 1}} \circ \ds h = \ds{f}
    \qquad
    \ds{\pi_{\mathbf 2}} \circ \ds h = \ds{g}
  \]
  
  Then, for every $i$, applying Lemma~\ref{lem:cpow-composition} we obtain
  \begin{align*}
    \pr_i (\ds{\pi_{\mathbf 1}} \circ \ds h) &= \pr_i \ds{\pi_{\mathbf 1}} \circ \pr_0 \A{D}^i\ds{h}\\
                                             &= \pi_1 \circ \cpow{i} \circ \pr_0 \A{D}^i\ds{h}\\
                                             &= \pi_1 \circ h_i\\
                                             &= f_i
  \end{align*}
  Applying a similar reasoning to $g_i$, and by the universal property of $\pair{f_i}{g_i}$ in $\cat{C}$, we obtain that $h_i = \pair{f_i}{g_i} = \pr_i \pair{\ds f}{\ds g}$ and hence $\ds h = \pair{\ds f}{\ds g}$. Therefore, $\ds{A} \times \ds{B}$ is the categorical product in $\ftor{\CAct_\omega}(\cat{C})$.

  We can construct a terminal $\omega$-change action $\ds{\top}$ by picking the terminal object of $\cat{C}$ at every level. This uniquely determines the entire structure of $\ds{\top}$, since the only possible choice for every morphism is the universal morphism $!$ in $\cat{C}$.
  \[
    \ds{\top} \defeq (\seq{\top}, \seq{\seq{!}}, \seq{\seq{!}}, \seq{!})
  \]
  Note that the $\omega$-sequences $\seq{!}$ are all $\omega$-differential maps, and hence $\ds{\top}$ is an $\omega$-change action.
  
  Now take an arbitrary $\omega$-change action $\ds{A}$. It is straightforward to check that there is exactly one morphism $\ds{!} : \ds{A} \ra \ds{\top}$, namely the morphism given by $\pr_i \ds{!} = !$. Therefore $\ds{\top}$ is the terminal object in $\ftor{\CAct_\omega}(\cat{C})$.
  
  We now sketch a proof that $\ftor{\CAct_\omega}(\cat{C})$ has exponentials, provided that $\cat{C}$ is cartesian closed and has all countable limits.
  First, consider $\omega$-sequences of $\cat{C}$-objects $\seq{A_i}$ and $\seq{B_i}$. Since $\cat{C}$ has all countable products, one can construct the infinite product
  \[
    \pr_j (\seq{A_i} \Ra \seq{B_i}) \defeq \pr_j \A{D}\seq{A_i} \Ra \pr_j \seq{B_i}
  \]
  Intuitively, this object of $\cat{C}$ represents the pre-$\omega$-differential maps between $\seq{A_i}$ and $\seq{B_i}$. \lo{I.e.~this object is $\prod_{j \in \omega} (\pr_j \A{D} \changeseq A \Ra B_j)$. But you don't seem to make use of this infinite product in the rest of the argument.

  OK - you build a limit object $|\ds A \Ra \ds B|$, but in the diagrams you only make use of $\pr_{j}(\changeseq A \Ra \changeseq B)$, not the infinite product.}

  If $\ds{A}, \ds{B}$ are $\omega$-change actions on the $\omega$-sequences $\seq{A_i}, \seq{B_i}$, we can consider the subobject of 
  $\pr_j (\seq{A_i} \Ra \seq{B_i}) \times \pr_{j+1} (\seq{A_i} \Ra \seq{B_i})$
  where the second element is the derivative of the first (i.e. of differential maps) by taking the limit of the following diagram:
  \begin{center}
    \begin{tikzcd}[ampersand replacement=\&]
      \&[-50]\pr_j (\seq{A_i} \Ra \seq{B_i}) \times \pr_{j+1} (\seq{A_i} \Ra \seq{B_i})
      \arrow[ddr, "\ulcorner \times \urcorner"]
      \arrow[ddl, swap, "\pi_1"]
      \\\\
      \pr_j \A{D}\seq{A_i} \Ra \pr_j \seq{B_i}
      \arrow[ddr, swap, "\pr_0\ds{\oplus}_j \Ra \Id"]
      \&\&[-90pt]
      (\pr_j \A{D}\seq{A_i} \times \pr_{j+1} \A{D}\seq{A_i}) \Ra (\pr_j \seq{B_i} \times \pr_{j + 1} \seq{B_i})
      \arrow[ddl,"\Id \Ra \pr_0 \ds{\oplus}_j"]
      \\\\
      \&(\pr_j \A{D}\seq{A_i} \times \pr_{j+1}\A{D}\seq{A_i}) \Ra \pr_j \seq{B_i}
    \end{tikzcd}
  \end{center}
  We can further restrict the space to only regular derivatives by taking the limit of a similar diagram, requiring that 
  \begin{align*}
  \d f(a, 0) &= 0\\
   \d f(a, \change a + \change b) &= \d f(a, \change a) + \d f (a + \change a, \change b).
  \end{align*}
  Pasting all these diagrams together, we can define the space of $\omega$-differential maps between $\ds{A}$ and $\ds{B}$ as a limit object
  $\us{\ds{A} \Ra \ds{B}}$ internal to $\cat{C}$.

  The $\omega$-sequence $\seq{\us{\ds{A} \Ra \Pi^i\ds{B}}}$ is then a pre-$\omega$-change action that forms the basis for the exponential $\ds{A} \Ra \ds{B}$ in $\ftor{\CAct_\omega}(\cat{C})$.
  \lo{The preceding sentence is unclear. 
  I read you as: The $\omega$-sequence $\seq{\us{\ds{A} \Ra \Pi^i\ds{B}}}$ of $\cat{C}$-objects is the first component of the $\omega$-change action $\ds{A} \Ra \ds{B}$ you wish to construct. Yes?}
  The structure morphisms, $\ds{\oplus}, \ds{+}$ and $\ds{0}$, are obtaining by lifting the structure morphisms in $\Pi^i\ds{B}$ pointwise.
  \map{To finish the proof in detail would require a lot of space and I am not sure whether it is very enlightening - the main trick has already been explained.}  
  \lo{This result is important and interesting, and so, well-worth writing out clearly. 
  This last paragraph need filling out. We also need to verify that the construction is the exponential object, i.e.~checking the universal properties.}
  
\end{proof-for-theorem}

\begin{proof-for-theorem}{cact-omega-canonical}
Given an object $\ds{A} = \left(\seq{A_i}, \seq{\ds{\oplus}_i}, \seq{\ds{+}_i}, \seq{0_i}\right)$ of $\ftor{\CAct_\omega}(\cat{C})$,
the canonical coalgebra $\gamma : \ftor{\CAct_\omega}(\cat{C}) \ra \ftor{\CAct}(\ftor{\CAct_\omega}(\cat{C}))$ maps the $\omega$-change action to itself. 
That $\ds{A}$ is a internal change action of $\ftor{\CAct_\omega}(\cat{C})$ follows at once from the definition of $\omega$-change action:
$\changes \ds A = \Pi \ds A$ and $\pr_n (\A{D} (\ds A \times \Pi {\ds A})) = \pr_{n+1} (\A{D} \ds A)$; and
$\ds{\oplus}_0 : \ds A \times \changes{\ds A} \to \ds A$ and $\ds{+}_0 : \changes{\ds A} \times \changes{\ds A} \to \changes{\ds A}$ are $\omega$-differential.
The functor $\gamma$ maps an $\omega$-differential map $\ds{f} : \ds A \to \ds B$ to the differential map $\gamma(\ds{f}) \defeq  (\ds f, \d \ds f)$ where $\d \ds f : \ds A \times \changes{\ds A} \to \changes{\ds B}$ is just the $\omega$-differential map $\Pi \ds{f}$.

\end{proof-for-theorem}

\begin{proof-for-theorem}{cact-as-limit}
  First we construct by induction on $i \geq 1$ a family of forgetful functors 
  $\epsilon_i : \ftor{\CAct_\omega}(\cat{C}) \ra \ftor{\CAct}^i(\cat{C})$ that make the diagram commute.
  
  When $i = 1$, we use the operator $\changes(\ds{A}, i)$ from Definition~\ref{def:pre-omega-change-action} to define $\epsilon_1$:
  \begin{itemize}
    \item $\epsilon_1(\ds{A}) \defeq \changes(\ds{A}, 1)$
    \item $\epsilon_1(\ds{f}) \defeq (\pr_0 \ds{f}, \pr_1 \ds{f})$ 
  \end{itemize}
  
  For $i \geq 0$, the functor $\epsilon_{i + 1}$ is defined inductively by:
  \begin{itemize}
    \item $\epsilon_{i + 1}(\ds{A}) \defeq (\epsilon_i(\ds{A}), \epsilon_i(\Pi\ds{A}), \epsilon_{i + 1}(\ds{\oplus_0^A}), \epsilon_{i + 1}(\ds{+_0^A}), \pr_0 0^A)$
    \item $\epsilon_{i + 1}(\ds{f}) \defeq (\epsilon_i(\ds{f}), \epsilon_i(\Pi \ds{f}))$
  \end{itemize}

  It is straightforward to check that the required diagram does commute. For example, for $i = 1$, $\epsilon_{2}(\ds{A})$ is the change action
  \begin{align*}
    \epsilon_2(\ds{A}) &= (A_{01}, A_{12}, \oplus_{012}, \ldots)\\
    A_{01} &= (A_0, A_1, (\pr_0\oplus_0, \pr_1 \oplus_0), \ldots)\\
    A_{12} &= (A_1, A_2, (\pr_0\oplus_1, \pr_1 \oplus_1), \ldots)\\
    \oplus_{012} &= ((\pr_0 \oplus_0, \pr_1 \oplus_1), (\pr_1 \oplus_1, \pr_2 \oplus_2))
  \end{align*}
  hence the ``lower'' structure extracted by $\epsilon$ and the ``higher'' one extracted by $\xi$ coincide.

  To prove the universal property, consider a category $\cat{D}$ and functors $\epsilon'_i : \cat{D} \ra \ftor{\CAct}^i(\cat{C})$ making the diagram commute.
  Then there is a unique functor $\seq{\epsilon'_i} : \cat{D} \ra \ftor{\CAct_\omega}(\cat{C})$ satisfying $\epsilon'_i = \epsilon_i \circ \seq{\epsilon'_i}$. 
  To construct it, first consider an object $U$ of $\cat{D}$. We define the $\omega$-sequence $\seq{U_i}$ by:
  \begin{align*}
    U_0 &\defeq \us{\epsilon'_1(U)}\\
    U_{j + 1} &\defeq \changes^{j + 1}\epsilon'_{j + 1}(U)
  \end{align*}
  Note that, for every $j$, $\epsilon'_{j + 1}(U)$ is a change action on $\ftor{\CAct}^{j + 1}(\cat{C})$ and, therefore, $\changes^j\epsilon'_{j+1}(U)$ is a change action
  in $\ftor{\CAct}(\cat{C})$. In particular, $\oplus_{\changes^j\epsilon'_{j + 1}}$ is an action of $U_{j + 1}$ on $U_j$. Hence we define the pre-$\omega$-differential map
  $\ds{\oplus^U_j}$ as follows:
  \begin{gather*}
    \pr_k \ds{\oplus^U_j} = \d^k \oplus_{\changes^j\epsilon'_{j + k + 1}}
  \end{gather*}
  and similarly for $\ds{+^U_j}, \ds{0^U_j}$.
  
  Then the action of $\seq{\epsilon'_i}$ on an object $U$ of $\cat{D}$ can be defined as:
  \[
    \seq{\epsilon'_i}(U) = (\seq{U_j}, \seq{\ds{\oplus^U_j}}, \seq{\ds{+^U_j}}, \seq{0^U})
  \]
  Note that this is indeed an $\omega$-change action, since the maps $\ds{\oplus^U_j}, \ds{+^U_j}$ are $\omega$-differential and satisfy the required equations by construction.

  Whenever $f : U \ra V$ is a morphism in $\cat{D}$, the morphism $\epsilon'_{j + 1}(f)$ is a differential map in $\ftor{\CAct}^{j + 1}(\cat{C})$, hence
  its $j$-th derivative $\d^{j + 1} f$ is a morphism in $\cat{C}$ of the appropriate type, so we can express the action of $\seq{\epsilon'_i}$ on morphisms of $U$ by:
  \[
    \pr_0\seq{\epsilon'_i}(f) = \us{\epsilon'_{1} (f)}\\
    \pr_{j + 1}\seq{\epsilon'_i}(f) = \d^{j + 1} \epsilon'_{j + 1}(f)
  \]
  which is an $\omega$-differential morphism by construction.

  \begin{rem}
    Note that the previous statement depends on $\epsilon'_{j + 1}$ equalising $\epsilon$ and $\xi$.
    If this were not the case, then we could have $\epsilon'_2 (f) = ((f_0, f_1), (f_1', f_2))$ with $f_1 \neq f_1'$.
    Then, according to the above definition:
    \begin{gather*}
      \pr_0 \seq{\epsilon'_i}(f) = f_0\\
      \pr_1 \seq{\epsilon'_i}(f) = f_1\\
      \pr_2 \seq{\epsilon'_i}(f) = f_2
    \end{gather*}
    However, since $f_1 \neq f_1'$, there is no guarantee that $f_2$ is a derivative for $f_1$, hence $\seq{\epsilon'_i}(f)$ is not $\omega$-differential.
  \end{rem}

  Thus defined, $\seq{\epsilon'_i}$ is a functor from $\cat{D}$ into $\ftor{\CAct_\omega}(\cat{C})$ such that $\epsilon'_i = \epsilon_i \circ \seq{\epsilon'_i}$,
  and it is clear from the construction that it is unique. Therefore, $\ftor{\CAct_\omega}(\cat{C})$ is precisely the desired limit.
   
\end{proof-for-theorem}
\section{Change actions as 2-categories}
\label{sec:2-categories}
Consider a change action $\da{A}$ in $\cat{Set}$. 
The change action induces the structure of a category $\ftor{Cat}(\da A)$ on $\us A$ as follows:
\begin{itemize}
  \item The objects of $\ftor{Cat}(\da A)$ are the elements of $\us A$.
  \item The morphisms $\ftor{Cat}(\da A)(a_1, a_2)$ are the changes 
    $\change a : \changes A$
    such that $a_1 \oplus \change a = a_2$.
  \item The identity morphism $\Id_A$ is the object $0 : \changes A$.
  \item Composition $\change a_2 \circ \change a_1$ is the sum $\change a_1 + \change a_2$.
\end{itemize}
Since $\oplus$ is a monoid action, the composition of morphisms is well-typed.
Associativity and identity follow from the fact that $\changes A$ is a monoid.

Now let $\da{f} = (\us f, \d f)$ be a differential map between change actions
$\da{A}$ and $\da{B}$. Clearly $\da{f}$ can be seen as a functor $\ftor{Cat}(\da{f})$
between the corresponding categories $\ftor{Cat}(\da A), \ftor{Cat}(\da B)$ 
in the following way:
\begin{itemize}
  \item If $a$ is an object of $\ftor{Cat}(\da A)$, then $\ftor{Cat}(f)(a)$ is
    the element $\us{f}(a)$ considered as an object of $\ftor{Cat}(\da B)$.
  \item If $\change a$ is a morphism from $a_1$ into $a_2$, then
    $\ftor{Cat}(f)(\change a)$ is $\d f(a_1, \change a)$.
\end{itemize}
This definition is well-typed since $\d f(a_1, \change a)$ is a change mapping
$\us{f}(a_1)$ into $\us{f}(a_1 \oplus \change a) = \us{f}(a_2)$, and hence a morphism from 
$\us{f}(a_1)$ into $\us{f}(a_2)$ in the category $\ftor{Cat}(\da B)$.
Functoriality follows from (and is equivalent to) regularity of $\d f$.

Conversely, let $\ftor{F}$ be a functor from $\ftor{Cat}(\da A)$ into
$\ftor{Cat}(\da B)$. 
This induces a differential map $\da{\ftor{F}} = (F, \d F)$ from $\da{A}$ into $\da{B}$ defined by:
\begin{align*}
  F(a) &\defeq \ftor{F}(a)\\
  \d F(a, \change a) &\defeq \ftor{F}(\change a)
\end{align*}
Regularity follows from functoriality of $\ftor{F}$, and the derivative property
is a direct consequence of the fact that $\ftor{F}$ is well-typed.

\begin{lem}
  The category $\cat{\CAct}$ embeds fully and faithfully into the 2-category
  $\cat{Cat}$ of (small) categories and functors.
\end{lem}

Given differential maps $\da{f}, \da{g} : \da{A} \to \da{B}$, 
a natural transformation $\vf{U} : \da{f} \xrightarrow{\cdot} \da{g}$ maps every object $a : \us{A}$ to a change
$\vf{U}(a) : \changes B$ such that the following diagram commutes:
\begin{center}
  \begin{tikzcd}
    \us{f}(a)
    \arrow[r, "{\d f(a, \change a)}"]
    \arrow[d, swap, "\vf{U}(a)"]
    & \us{f}(a \oplus \change{a})
    \arrow[d, "\vf{U}(a \oplus \change a)"]
    \\
    \us{g}(a)
    \arrow[r, "{\d g(a, \change a)}"]
    & \us{g}(a \oplus \change{a})
  \end{tikzcd}
\end{center}
In particular, this means natural transformations are a subset of the set 
$\us{A} \ra \changes{B}$, which can be read as generalized vector fields (mapping
the space $A$ to $\changes B$ rather than $\changes A$).

\begin{rem}
  Consider a natural transformation from functor $\ftor{Cat}(\da f)$ into
  $\ftor{Cat}(\da g)$. This is, first and foremost, a map that assigns to
  every element $a \in \us A$ a change $\change a \in \changes A$. This is precisely
  the space of functional changes $\changes (A \Ra B)$ in $\cat{\CAct_\omega}$
  (see Sec.~\ref{sec:omega-change-actions})
  and, in general, in any change action model (see Sec.~\ref{sec:extrinsic}) equipped with an infinitesimal object.
\end{rem}

More generally, the category $\ftor{\CAct}(\cat{C})$ of change actions on an arbitrary base cartesian category $\cat{C}$ can be regarded as a 2-category. 
Indeed, given change actions $\da{A}, \da{B}$ we define the category of differential maps $\cat{Diff}(\da{A}, \da{B})$ as follows:
\begin{itemize}
  \item The objects of $\cat{Diff}(\da{A}, \da{B})$ are differential maps
    $\da{f} : \da{A} \ra \da{B}$.
  \item The morphisms between $\da{f}, \da{g}$ are $\cat{C}$-morphisms 
    $\vf{U} : \us A \ra \changes B$ such that the following diagrams
    (in $\cat{C}$) commute:
    \begin{center}
      \begin{tikzcd}
        A
        \arrow[r, "\pair{f}{\vf{U}}"]
        \arrow[dr, swap, "g"]
        & B \times \changes{B} \arrow[d, "\oplus"]
        \\
        &B
      \end{tikzcd}
      \begin{tikzcd}
        A \times \changes{A}
        \arrow[r, "\pair{\d f}{\vf{U} \circ \oplus}"]
        \arrow[d, swap, "\pair{\vf{U} \circ \pi_1}{\d g}"]
        &\changes{B} \times \changes{B}
        \arrow[d, "+"]
        \\
        \changes{B} \times \changes{B}
        \arrow[r, "+"]
        &\changes{B}
      \end{tikzcd}
    \end{center}
    Intuitively, the first diagram asserts that $\vf{U}$ has the ``type''
    of a natural transformation $f$ to $g$, whereas the second diagram states 
    naturality of $\vf{U}$.
\end{itemize}
The identity objects in $\cat{Diff}(\da{A}, \da{B})$ are the constant zero maps
$\Id_{\da{f}} \defeq (0_{\da{B}}) \circ !$.
Given $\cat{Diff}(\da{A}, \da{B})$-morphisms $\vf{U} : \da{f} \ra \da{g}$,
$\vf{V} : \da{g} \ra \da{h}$, their  composition is defined by:
\begin{gather*}
  \vf{V} \bullet \vf{U} \defeq + \circ \pair{\vf{U}}{\vf{V}}
\end{gather*}
which is a $\cat{Diff}$-map between $\da{f}$ and $\da{h}$ - indeed:
\begin{align*}
  h &= \oplus \circ \pair{g}{\vf{V}}\\
    &= \oplus \circ \pair{\oplus \circ \pair{f}{\vf{U}}}{\vf{V}}\\
    &= \oplus \circ \pair{f}{+ \circ \pair{\vf{U}}{\vf{V}}}
\end{align*}
\begin{align*}
  + \circ \pair{\d f}{+ \circ \pair{\vf{U}}{\vf{V}} \circ \oplus}
  &= + \circ \pair{+ \circ \pair{\d f}{\vf{U} \circ \oplus}}{\vf{V} \circ \oplus}\\
  &= + \circ \pair{+ \circ \pair{\vf{U} \circ \pi_1}{\d g}}{\vf{V} \circ \oplus}\\
  &= + \circ \pair{\vf{U} \circ \pi_1}{+ \circ \pair{\d g}{\vf{V} \circ \oplus}}\\
  &= + \circ \pair{\vf{U} \circ \pi_1}{+ \circ \pair{\vf{V} \circ \pi_1}{\d g}}\\
  &= + \circ \pair{+ \circ \pair{\vf{U}}{\vf{V}} \circ \oplus}{\d g}
\end{align*}
which entails the required diagram commutes. Associativity and identity follow
from the definition of change action.

Furthermore, composition of differential maps can be lifted to a functor on
the corresponding categories. More precisely, let $\da{A}, \da{B}, \da{C}$ be
change actions. Define the functor $\ftor{Comp} : \cat{Diff}(\da{A}, \da{B})
\times \cat{Diff}(\da{B}, \da{C}) \ra \cat{Diff}(\da{A}, \da{C})$ as follows:
\begin{itemize}
  \item If $(\da{f}, \da{g})$ is an object in $\cat{Diff}(\da{A}, \da{B}) \times
    \cat{Diff}(\da{B}, \da{C})$, then $\ftor{Comp}(\da{f}, \da{g})$ is
    just the composition of differential maps, i.e.
    \begin{gather*}
      \ftor{Comp}(\da{f}, \da{g}) = (\us g \circ \us f, \d g \circ \pair{\us f \circ \pi_1}{\d f})
    \end{gather*}
  \item If $(\vf{U}, \vf{V})$ is a morphism from $(\da{f_1}, \da{g_1})$ into
    $(\da{f_2}, \da{g_2})$, then $\ftor{Comp}(\vf{U}, \vf{V})$ is defined as:
    \begin{gather*}
      + \circ \pair{\d g_1 \circ \pair{f_1}{\vf{U}}}{\vf{V} \circ \us{f_2}}
    \end{gather*}
\end{itemize}
\begin{proof}
  We need to show that, as defined above, the morphism $\ftor{Comp}(\vf{U}, \vf{V})$
  is indeed a 2-cell (i.e. the required diagrams commute).
  
  It's easy to verify the first. Indeed:
  \begin{align*}
    g_2 \circ f_2 &= \oplus \circ \pair{g_1}{\vf{V}} \circ f_2\\
                  &= \oplus \circ \pair{g_1 \circ f_2}{\vf{V} \circ f_2}\\
                  &= \oplus \circ \pair{g_1 \circ \oplus \circ \pair{f_1}{\vf{U}}}{\vf{V} \circ f_2}\\
                  &= \oplus \circ 
                    \pair{\oplus \circ \pair{g_1 \circ f_1}{\d g_1 \circ \pair{f_1}{\vf{U}}}}
                    {\vf{V} \circ f_2}\\
                  &= \oplus \circ \pair{g_1 \circ f_1}
                    {+ \circ \pair{\d g_1 \circ \pair{f_1}{\vf{U}}}{\vf{V} \circ f_2}}\\
                  &= \oplus \circ \pair{g_1 \circ f_1}
                                       {\ftor{Comp}(\vf{U}, \vf{V})}
  \end{align*}
  For the second:
  \begin{align*}
    + \circ &\pair{\ftor{Comp}(\vf{U}, \vf{V}) \circ \pi_1}{\d (g_2 \circ f_2)}\\
    &= + \circ \pair{+ \circ \pair{\d g_1 \circ \pair{f_1}{\vf{U}}}{\vf{V} \circ f_2} \circ \pi_1}
    {\d g_2 \circ \pair{f_2 \circ \pi_1}{\d f_1}}\\
    &= + \circ \pair{\d g_1 \circ \pair{f_1}{\vf{U}} \circ \pi_1}
      {+ \circ \pair{\vf{V} \circ f_2 \circ \pi_1}{\d g_2 \circ \pair{f_2 \circ \pi_1}{\d f_2}}}\\
    &= + \circ \pair{\d g_1 \circ \pair{f_1}{\vf{U}} \circ \pi_1}
      {+ \circ \pair{\vf{V} \circ f_2 \circ \pi_1}{\d g_2 \circ \pair{f_2 \circ \pi_1}{\d f_2}}}\\
    &= + \circ \pair{\d g_1 \circ \pair{f_1}{\vf{U}} \circ \pi_1}
      {+ \circ \pair{\vf{V} \circ \pi_1}{\d g_2} \circ ((f_2 \circ \pi_1) \times \d f_2)}\\
    &= + \circ \pair{\d g_1 \circ \pair{f_1}{\vf{U}} \circ \pi_1}
      {+ \circ \pair{\d g_1}{\vf{V} \circ \oplus} \circ ((f_2 \circ \pi_1) \times \d f_2)}\\
    &= + \circ \pair{+ \circ \pair{\d g_1 \circ \pair{f_1}{\vf{U}} \circ \pi_1}
                                  {\d g_1 \circ ((f_2 \circ \pi_1) \times \d f_2)}}
                    {\vf{V} \circ \oplus \circ ((f_2 \circ \pi_1) \times \d f_2)}\\
    &= + \circ \pair{+ \circ \pair{\d g_1 \circ \pair{f_1}{\vf{U}} \circ \pi_1}
                                  {\d g_1 \circ (((+ \circ \pair{f_1}{\vf{U}}) \circ \pi_1) \times \d f_2)}}
                    {\vf{V} \circ \oplus \circ ((f_2 \circ \pi_1) \times \d f_2)}\\
   & \text{(by regularity of $\d g_1$)}\\
   &= + \circ \pair{\d g_1 \circ \pair{f_1 \circ \pi_1}{+ \circ \pair{\vf{U} \circ \pi_1}{\d f_2}}}
                   {\vf{V} \circ \oplus \circ ((f_2 \circ \pi_1) \times \d f_2)}\\
   &= + \circ \pair{\d g_1 \circ \pair{f_1 \circ \pi_1}{+ \circ \pair{\d f_1}{\vf{U} \circ \oplus}}}
                   {\vf{V} \circ \oplus \circ ((f_2 \circ \pi_1) \times \d f_2)}\\
   & \text{(by regularity of $\d g_1$ and reassociating)}\\
   &= + \circ \pair{\d g_1 \circ \pair{f_1 \circ \pi_1}{\d f_1}}
                   {+ \circ \pair{\d g_1 \circ \pair{\oplus \circ \pair{f_1 \circ \pi_1}{\d f_1}}{\vf{U}\circ \oplus}}
                                 {\vf{V} \circ \oplus \circ ((f_2 \circ \pi_1) \times \d f_2)}}\\
   &= + \circ \pair{\d (g_1 \circ f_1)}
                   {+ \circ \pair{\d g_1 \circ \pair{f_1 \circ \oplus}{\vf{U}\circ \oplus}}
                                 {\vf{V} \circ f_2 \circ \oplus}}\\
   &= + \circ \pair{\d (g_1 \circ f_1)}
                   {+ \circ \pair{\d g_1 \circ \pair{f_1}{\vf{U}}}
                                 {\vf{V} \circ f_2}
                      \circ \oplus}\\
   &= + \circ \pair{\d (g_1 \circ f_1)}
                   {\ftor{Comp}(\vf{U}, \vf{V}) \circ \oplus}
  \end{align*}
  
  That $\ftor{Comp}$ is a functor follows straightforwardly from monoidality.
  \qed 
\end{proof}

\begin{thm}
  The category $\cat{\CAct}(\cat{C})$ of change actions on a base category $\cat{C}$ is a 2-category with the structure described above.
\end{thm}
\end{document}